\documentclass[11pt,oneside]{article}

\usepackage{amsmath,amssymb,amsthm,bm}
\usepackage{geometry}
\usepackage{mathrsfs}
\usepackage[scr=boondox]{mathalpha}
\usepackage{booktabs}
\usepackage{comment}
\usepackage[english]{babel}
\usepackage{cite}
\usepackage[hidelinks]{hyperref}

\numberwithin{equation}{section}

\newtheorem{theorem}{Theorem}[section]
\newtheorem{proposition}[theorem]{Proposition}

\newtheorem*{remark}{Remark}

\title{Fock-Space Representation of the Lebowitz--Frisch--Helfand Kinetic Model}

\author{Ilya Karlin\thanks{ikarlin@ethz.ch}\\
Department of Mechanical and Process Engineering\\
ETH Zurich, CH-8092 Zurich, Switzerland}

\date{\today}

\begin{document}
\maketitle

\begin{abstract}
We develop an exact Fock-space representation of the 
Lebowitz--Frisch--Helfand kinetic equation.  The familiar square-root
Maxwellian transformation is used only as a convenient starting point: it
maps the local Ornstein--Uhlenbeck relaxation sector to the bosonic number
operator and thereby exposes an elementary grading.  The main issue is the
representation of the complete kinetic dynamics when the velocity
realization depends on local macroscopic parameters.  We formulate the
pull-back through an arbitrary admissible realization map, show that
external space--time derivatives acquire a differential connection, and
prove an intertwining theorem for first-order propagation operators.  The
physical velocity moments are represented by dual Fock-space functionals;
a transport--moment intertwining theorem then evaluates the complete
propagation contribution without requiring the explicit differential
connection.  Compatibility between realization parameters and kinetic
moments may be imposed algebraically or propagated dynamically by exact
balance laws.  In the hydrodynamic limit, the number grading selects the
second and third Fock levels responsible for viscous stress and heat flux,
leading directly to the Navier--Stokes--Fourier constitutive terms.  Finally,
we prove covariance under local changes of coordinate realization and
illustrate it explicitly for Hermite-polynomial and Hermite-function
coordinates.  Thus the Hermite realization is computationally privileged,
while the intertwined Fock dynamics, physical moments, and compatibility
structure are representation-independent within the admissible similarity
class.
\end{abstract}

\section{Introduction}
\label{sec:introduction}

The Lebowitz--Frisch--Helfand (LFH) model is a nonlinear kinetic equation in
which the velocity-space relaxation is of Ornstein--Uhlenbeck type, while its
local drift velocity and temperature are determined self-consistently by the
moments of the evolving distribution \cite{LebowitzFrischHelfand1960}.  It is
therefore unusually transparent: for prescribed local parameters the kinetic
generator is linear, whereas the nonlinearity resides in the adaptation of
those parameters to the represented state.  This separation makes the LFH
model a useful test case for asking a representation-theoretic question that
is largely independent of the complexity of the collision model itself.

The similarity between Fokker--Planck or Kramers generators and
imaginary-time Schr\"odinger operators, as well as their oscillator and
Hermite representations, is classical \cite{UhlenbeckOrnstein1930,Kramers1940,Risken1989}.
Likewise, Hermite expansions are standard tools of kinetic theory
\cite{GradHermite1949,GradKinetic1949,ChapmanCowling1970,Liboff2003}.  None of
these facts is a novelty claim of the present work.  We use the square root
of the local Maxwellian deliberately because, for the LFH relaxation sector,
it removes the first-order velocity derivative and converts the relaxation
contribution into the number operator.  This gives the local sector an
immediate grading and allows the remainder of the kinetic structure to be
examined without additional collision-operator complications.

The bosonic Fock construction and the associated creation and annihilation
operators are equally classical \cite{Fock1932,Berezin1966,Hall2013}; the
Weyl/oscillator representation used below may also be viewed in the standard
phase-space framework \cite{Folland1989}.  Fock-space and second-quantized
formalisms have long been used outside quantum mechanics, notably for
classical stochastic many-particle and reaction processes
\cite{Doi1976,Peliti1985}.  Our use is different in emphasis.  We do not
second-quantize a particle number.  Instead, the occupation-number grading is
used to represent the velocity-space hierarchy of a one-particle kinetic
state, while the main problem is to intertwine the \emph{complete} kinetic
dynamics, including external derivatives acting through locally
parameterized coordinate realizations.

That last point is essential.  Intertwining a polynomial in velocity and
velocity derivatives is an algebraic matter once a realization of the Weyl
algebra has been fixed.  The streaming operator is different.  The local
coordinate basis depends on the fields that translate and dilate the
Maxwellian, hence on space and time.  Pulling back an external derivative
therefore differentiates both the Fock coefficients and the realization map.
The resulting differential connection is not an optional correction; it is
the additional structure required for the realization to commute with
propagation.  We first derive this connection for an arbitrary admissible
realization and only afterwards evaluate it in Hermite-function coordinates.
This order is important: Hermite functions provide a convenient coordinate
calculation, not the definition of the abstract object.

The payoff of the construction is that, once the intertwining has been
established, the kinetic problem no longer has to be organized around a
particular special-function calculus.  Relaxation, propagation, physical
moments, compatibility, grading, and covariance are encoded in the abstract
operator structure.  Hermite functions are used to construct one convenient
coordinate realization and to evaluate its differential connectors, but the
subsequent relations among these structures do not require returning to
Hermite recurrences, Gaussian identities, or operations on the special
functions themselves.  In this sense the Fock formulation isolates the
representation-independent algebraic content of the LFH model from the
coordinate machinery used to calculate one representative.

Operationally, this also separates two tasks that are entangled in the usual
self-consistent formulation.  For prescribed local fields, the kinetic
dynamics is represented at the operator level; the identification of those
fields with moments of the represented physical distribution is then a
distinct compatibility problem.  Enforcing compatibility algebraically, or
propagating it dynamically, therefore does not require reconstructing the
velocity-coordinate representation.  This separation, rather than the
oscillator analogy itself, is the principal reason for introducing Fock space
here.

A second ingredient is the representation of observables.  Physical moments
are linear functionals of the distribution and transform contragrediently to
the coordinate state.  Once their Fock-space representatives are constructed,
the intertwining of the complete propagation operator yields a useful
transport--moment theorem: every velocity moment of transport can be written
in conservative form without ever inserting the explicit differential
connection.  This observation emerged from the explicit hydrodynamic
projections and is one of the structural simplifications emphasized below.

The local parameters and the represented kinetic moments are initially kept
independent.  Compatibility is the condition that identifies them.  For the
LFH model this condition can be implemented in two equivalent ways.  It may
be imposed algebraically at each time, in which case the exact moment
balances follow from the parameterized kinetic relation; or it may be imposed
initially and propagated dynamically by evolving the local parameters with
those exact balances.  The second formulation is particularly useful for the
Fock representation because it treats the moving local realization as part
of the dynamics and suggests a route that need not rely on an explicitly
solvable moment fixed point in more general applications.

The representation is finally required to be covariant.  We distinguish the
square-root Maxwellian transformation, which is a computationally privileged
entry point, from the subsequent freedom to choose coordinates for the same
abstract Fock state.  Two local realization maps related by an invertible
similarity transformation generally have different coordinate vacua,
differential operators, external-derivative connections, and moment weights.
When these objects are transformed together, however, they pull back to the
same Fock generator and the same physical moment functionals.  The explicit
comparison of Hermite-polynomial and Hermite-function coordinates provides a
nontrivial example of this covariance.

Related nonlinear Fokker--Planck kinetic models and stochastic-particle
realizations have been developed for rarefied-gas computations
\cite{JennyTorrilhonHeinz2010,GorjiJenny2014,GorjiTorrilhon2021}.  More
broadly, linear embeddings of nonlinear dynamics have also acquired renewed
interest in quantum algorithms 
\cite{LiuEtAl2021,bravyi2025quantumsimulationnoisyclassical,bravyi2026quantumalgorithmsstochasticnonlinear}.  The present paper does not propose a quantum algorithm;
the relevance of the Fock formulation here is structural: it provides an
exact operator representation, a natural grading, and a coordinate-covariant way to organize transport and moment projections.

The paper is organized as follows.  Section~\ref{sec:LFH_equation} formulates
the parameterized LFH operator, its macroscopic moments and compatibility,
and performs the square-root Maxwellian transformation.  Section~\ref{sec:fock_space} introduces
the abstract Fock space and two coordinate realizations.  Section~\ref{sec:intertwining} develops
the general intertwining construction, with special attention to external
derivatives and the induced differential connection, and then obtains the
Fock-space LFH generator.  Section~\ref{sec:fock_compatibility} constructs physical moment functionals,
proves the transport--moment intertwining theorem, and establishes exact
compatibility propagation.  Section~\ref{sec:Chapman_Enskog} uses the number
grading to obtain the hydrodynamic limit.  Section~\ref{sec:covariance}
proves covariance under changes of coordinate realization and gives the
explicit polynomial/function comparison. Results are discussed in Section~\ref{sec:conclusion}. A detailed independent derivation
of the Hermite-function differential connection is collected in the
Appendix~\ref{app:hermite_differential_connectors}.

\section{The Lebowitz--Frisch--Helfand Kinetic Equation}
\label{sec:LFH_equation}

We begin from a parameterized family of LFH operators rather than from the
self-consistent kinetic equation.  This distinction keeps separate the
operator structure, the physical moment functionals, and the compatibility
condition that ultimately couples them.  The same separation will later be
carried through the Fock representation.

\subsection{The LFH Operator}
Throughout this work we employ Einstein's summation convention over
repeated Cartesian indices. Let us denote by $\mathcal L$ the Lebowitz--Frisch--Helfand (LFH) operator in $d$ dimensions, acting on distribution functions $f(\bm v;\bm x,t)$,
\begin{align}
    \mathcal L&=\mathcal T-\gamma\mathcal{L}_{\rm OU}.\label{eq:LFH_op}
\end{align}
Here $\mathcal T$ is the propagation operator, partial derivatives in time and space are denoted as $\partial_\mu$, $\mu\in\{t,1,\dots,d\}$,
\begin{align}
    \mathcal T&=\partial_t+v_\alpha\partial_\alpha.\label{eq:Stream_op}
\end{align}
Furthermore, 
$\gamma$ is the inverse relaxation time, and $\mathcal L_{\rm OU}$ is the Ornstein--Uhlenbeck (OU) operator describing the drift and the diffusion terms in the velocity-space,
\begin{align}
    \mathcal{L}_{{\rm OU}}\Box&=\frac{\partial}{\partial v_\alpha}
\left[
(v_\alpha-u_\alpha)\Box
\right]
+\frac{\partial^2}
{\partial v_\alpha\partial v_\alpha}
\left(
\theta
\Box\right).\label{eq:LOU_def}
\end{align}
 We consider throughout the force-free LFH model in order to keep the focus on propagation and relaxation.
The OU operators \eqref{eq:LOU_def} constitute a family parameterized by the local velocity parameter $\bm{u}(\bm x,t)$ and $\theta(\bm x,t)$. 
This parameterization is tied to the reference local Maxwellian,
\begin{equation}
 M
=
\frac{n}
{(2\pi\theta)^{d/2}}
\exp
\left(
-
\frac{
(v_\alpha-u_\alpha)
(v_\alpha-u_\alpha)
}
{2\theta}
\right),
\label{eq:MB}
\end{equation}
which is annihilated by the Ornstein--Uhlenbeck operator, 
\begin{equation}\label{eq:zero_LOU}
\mathcal L_{{\rm OU}}M=0.
\end{equation}
%
Introducing the peculiar velocity,
\begin{equation}
c_\alpha
=
v_\alpha-u_\alpha,
\label{eq:peculiar}
\end{equation}
the local Maxwellian assumes the compact form, $M=n W$, with the normalized weight $W$,

\begin{equation}\label{eq:W_maxwell}
W=\frac{1}{(2\pi\theta)^{d/2}}\exp\left(-\frac{c_\alpha c_\alpha}{2\theta}\right).
\end{equation}

\subsection{Macroscopic Fields}
\label{sec:macro_fields}

For a distribution $f$, define the physical number, momentum, and total-energy
densities independently of the parameters entering the operator:
\begin{align}
N[f]
&:=\int_{\mathbb R^d} f\,d^dv,
\label{eq:physical_number_classical}
\\
P_\alpha[f]
&:=m\int_{\mathbb R^d}v_\alpha f\,d^dv,
\label{eq:physical_momentum_classical}
\\
E[f]
&:=\frac m2\int_{\mathbb R^d}v_\alpha v_\alpha f\,d^dv.
\label{eq:physical_energy_classical}
\end{align}
These moments determine the physical mean velocity and
kinetic temperature,
\begin{equation}
 u_\alpha^{[f]}=\frac{P_\alpha}{mN},
 \qquad
 \theta^{[f]}
 =\frac{2}{mdN}
 \left(
 E-\frac{P_\alpha P_\alpha}{2mN}
 \right).
\label{eq:moment_derived_parameters}
\end{equation}
Thus the parameter triple
$\lambda=(n,\bm u,\theta)$ entering the local Maxwellian and the physical
moment triple extracted from $f$ are, at this stage, distinct objects.
For later use we also introduce the number flux, momentum flux, and total
energy flux,
\begin{align}
J_\alpha[f]
&:=\int v_\alpha f\,d^dv=\frac{P_\alpha}{m},
\\
\Pi_{\alpha\beta}[f]
&:=m\int v_\alpha v_\beta f\,d^dv,
\\
Q_\alpha[f]
&:=\frac m2\int v_\beta v_\beta v_\alpha f\,d^dv.
\label{eq:classical_fluxes}
\end{align}

\subsection{Compatibility and the LFH Equation}
\label{sec:classical_compatibility}

For prescribed smooth fields $\lambda=(n,\bm u,\theta)$, the relation
\begin{equation}
\mathcal L[\lambda]f=0
\label{eq:parameterized_LFH_relation}
\end{equation}
is a linear kinetic equation with a locally adapted Ornstein--Uhlenbeck
operator.  The nonlinear LFH equation is obtained when the parameters are
made compatible with the physical moments of the same distribution.  The
algebraic compatibility condition is
\begin{equation}
N=n,
\qquad
P_\alpha=mnu_\alpha,
\qquad
E=E_\lambda
:=\frac m2nu_\alpha u_\alpha+\frac{md}{2}n\theta .
\label{eq:classical_compatibility_fixedpoint}
\end{equation}
Equivalently, $\lambda$ is the fixed point of the moment map
$\lambda\mapsto(N,u^{[f]},\theta^{[f]})$.  Only after
Eq.~\eqref{eq:classical_compatibility_fixedpoint} is imposed do we refer to
Eq.~\eqref{eq:parameterized_LFH_relation} as the self-consistent LFH kinetic
equation.

There is an equivalent dynamical formulation.  Taking the first three
physical moments of the parameterized relation gives
\begin{align}
\partial_tN+\partial_\alpha J_\alpha
&=0,
\label{eq:classical_offcompat_density}
\\
\partial_tP_\beta+\partial_\alpha\Pi_{\beta\alpha}
+\gamma\left(P_\beta-mu_\beta N\right)
&=0,
\label{eq:classical_offcompat_momentum}
\\
\partial_tE+\partial_\alpha Q_\alpha
+\gamma\left(2E-u_\alpha P_\alpha-md\theta N\right)
&=0.
\label{eq:classical_offcompat_energy}
\end{align}
On the compatibility set the relaxation production terms vanish.  Their
vanishing by itself fixes the drift velocity and temperature relative to the
physical momentum and energy, but it does not enforce $N=n$, since the density
parameter does not enter the Ornstein--Uhlenbeck operator.  Hence algebraic
compatibility implies the exact conservative balance laws.  Conversely, on
compatible initial data one may evolve the local parameters with those
balance laws; the corresponding compatibility residuals then remain zero.
Thus the fixed-point and balance-law routes give equivalent implementations
of the LFH self-consistency on the compatibility manifold.  We shall
use the same distinction later in Fock space, where the dynamical route is
particularly natural because the local realization itself depends on
$\lambda$.

\subsection{Motivation for an Operator Formulation}

The parameterized formulation isolates the feature that makes the LFH model
useful here.  For fixed $\lambda$ the velocity-space generator is linear and
its equilibrium is the local Maxwellian $M_\lambda$, whereas the nonlinearity
enters only when $\lambda$ is adapted to the moments of $f$.  We therefore
seek an exact representation of the parameterized kinetic operator first and
postpone compatibility until the physical moment functionals have been
represented as well.

The square-root Maxwellian is chosen for a standard computational reason.  It
places multiplication by peculiar velocity and differentiation with respect
to velocity into a symmetric canonical pair and eliminates the first-order
velocity derivative from the Ornstein--Uhlenbeck sector.  This familiar
oscillator reduction is not the objective of the construction; it is the
entry point that makes the local relaxation sector elementary and leaves the
representation of propagation as the substantive problem.

\subsection{The Square-Root Maxwellian Similarity Transformation}\label{sec:LFH_transform_similarity}

\subsubsection{Definition of the Similarity Transformation}

We begin with an
exact similarity transformation that will serve as the basis for the
subsequent operator formulation.
Let $M$ denote the local Maxwellian \eqref{eq:MB} 
associated with the
hydrodynamic fields $n(\bm{x},t)$, $\bm u(\bm{x},t)$ and $\theta(\bm{x},t)$.
We introduce the transformed kinetic state
\begin{equation}
\label{eq:similarity}
f
=
M^{1/2}\,
\Phi.
\end{equation}
Equation (\ref{eq:similarity}) defines an exact and invertible
transformation between the distribution function $f$ and the transformed
state $\Phi$. With \eqref{eq:W_maxwell}, the square root of the Maxwellian assumes the form $M^{1/2}=\sqrt n\,\sqrt W$, where
\begin{equation}
\label{eq:sqrtM}
W^{1/2}=\frac{1}{(2\pi\theta)^{d/4}}\exp\!\left(-\frac{c_\alpha c_\alpha}{4\theta}\right).
\end{equation}
The transformation (\ref{eq:similarity}) transfers the Gaussian weight
from the kinetic distribution to the definition of the state itself. As
a consequence, the action of the LFH operator on $f$ is replaced by the
action of transformed differential operators on $\Phi$.

\subsubsection{Similarity Transformation of the Velocity Derivative}

The similarity transformation of the distribution function induces the corresponding similarity transformation of the LFH operator.
 We begin with derivatives with respect to the velocity
coordinates. 
Starting from the similarity transformation \eqref{eq:similarity}, 
the velocity derivative of the distribution function becomes
\begin{align}
\frac{\partial f}{\partial v_\alpha}=
M^{1/2}
\left(
\frac{\partial}{\partial v_\alpha}
-
\frac{c_\alpha}{2\theta}
\right)\Phi.
\label{eq:velocity_similarity_action}
\end{align}
Multiplying by $M^{-1/2}$ from the left yields
\begin{equation}
M^{-1/2}
\frac{\partial \left(M^{1/2}\Phi\right)}{\partial v_\alpha}
=
\left(
\frac{\partial}{\partial v_\alpha}
-
\frac{c_\alpha}{2\theta}
\right)\Phi.
\label{eq:velocity_similarity_action2}
\end{equation}
Since this identity holds for an arbitrary state $\Phi$, the
corresponding operator identity is
\begin{equation}
M^{-1/2}
\frac{\partial}
{\partial v_\alpha}
M^{1/2}
=
\frac{\partial}
{\partial v_\alpha}
-
\frac{c_\alpha}{2\theta}.
\label{eq:similarity_velocity}
\end{equation}
Equation (\ref{eq:similarity_velocity}) shows that ordinary
velocity differentiation is transformed into a shifted differential
operator consisting of the derivative itself together with a linear
multiplication operator.

\subsubsection{Similarity Transformation of the Streaming Operator}\label{sec:LFH_transform_dxdt}

We next consider differentiation with respect to the spatial and temporal coordinates.
Unlike the velocity derivative, the spatial and temporal derivative acts both on the
transformed kinetic state and on the hydrodynamic fields entering the
local Maxwellian. Consequently, the similarity transformation generates
additional terms involving the spatial gradients or time derivatives of the macroscopic
fields that parameterize the Maxwellian. 
For notational convenience, let
$
\partial_\mu
\in
\left\{
{\partial_t},
{\partial_\alpha}
\right\}$
denote a generic differential operator acting on the
space--time variables only. Since $\partial_\mu$ commutes with differentiation in
velocity space, the following derivation applies equally to temporal and
spatial differentiation.
Starting from the similarity transform \eqref{eq:similarity}, 
one finds
\begin{align}
\partial_\mu f
=
M^{1/2}
\left(
\partial_\mu
+
{M^{-1/2}}
\partial_\mu \left(M^{1/2}\right)\right)\Phi.
\end{align}
Hence,
\begin{equation}
M^{-1/2}
\partial_\mu
M^{1/2}
=
{\partial_\mu}
+
\Gamma_\mu,\ \mu\in\{t,1,\dots,d\},
\label{eq:Gamma_xt}
\end{equation}
where
\begin{equation}
\Gamma_\mu
=
\frac12
\partial_\mu
\ln M.
\label{eq:def_Gamma_xt}
\end{equation}
Using the explicit form of the Maxwellian, one obtains
\begin{align}
\Gamma_\mu
=
\frac{\partial_\mu n}{2n}
-\frac{d}{4\theta}\partial_\mu\theta
+\frac{c_\beta}{2\theta}\partial_\mu u_\beta
+\frac{c_\beta c_\beta}{4\theta^{2}}\partial_\mu\theta,
\qquad \mu\in\{t,1,\dots,d\}.
\label{eq:Gamma_xt_explicit}
\end{align}
%
We now apply the elementary similarity identities established in the
preceding subsections to the streaming part $\mathcal T$ of the LFH operator \eqref{eq:Stream_op}.
Since multiplication by the particle velocity commutes with the
similarity transformation,
\begin{equation}
M^{-1/2}v_\alpha M^{1/2}=v_\alpha,
\label{eq:v_commutes}
\end{equation}
only the differential operators are modified.
Using similarity identities \eqref{eq:Gamma_xt} and \eqref{eq:Gamma_xt_explicit}, we obtain
\begin{equation}
M^{-1/2}
\mathcal{T}
M^{1/2}
=
\frac{\partial}{\partial t}
+
v_\alpha
\frac{\partial}{\partial x_\alpha}
+
\Gamma_t
+
v_\alpha\Gamma_\alpha.
\label{eq:streaming_similarity_final}
\end{equation}

Equation~(\ref{eq:streaming_similarity_final}) exhibits the structure of
the transformed streaming operator in its simplest form. The original
kinetic transport operator is preserved, while the similarity
transformation generates additional contributions through the
space--time dependence of the local Maxwellian.

\subsubsection{Similarity Transformation of the Ornstein--Uhlenbeck Operator}

We finally consider the similarity transformation of the
Ornstein--Uhlenbeck relaxation operator \eqref{eq:LOU_def} which we write in the divergence form
\begin{equation}
\mathcal{L}_{OU}
=
\frac{\partial}{\partial v_\alpha}
\left(
c_\alpha
+
\theta
\frac{\partial}{\partial v_\alpha}
\right).
\label{eq:OU_operator}
\end{equation}

The similarity transformation is obtained by transforming first the
operator inside the parentheses and subsequently the outer velocity
derivative.

Since multiplication by the peculiar velocity commutes with the
similarity transformation, and taking into account \eqref{eq:similarity_velocity},
the inner operator becomes
\begin{align}
M^{-1/2}
\left(
c_\alpha
+
\theta
\frac{\partial}{\partial v_\alpha}
\right)
M^{1/2}
&=
c_\alpha
+
\theta
\left(
\frac{\partial}{\partial v_\alpha}
-
\frac{c_\alpha}{2\theta}
\right)
\nonumber\\
&=
\theta
\frac{\partial}{\partial v_\alpha}
+
\frac12 c_\alpha .
\label{eq:OU_inner}
\end{align}

Applying next the transformed outer derivative yields

\begin{align}
M^{-1/2}
\mathcal{L}_{OU}
M^{1/2}
&=
\left(
\frac{\partial}{\partial v_\alpha}
-
\frac{c_\alpha}{2\theta}
\right)
\left(
\theta
\frac{\partial}{\partial v_\alpha}
+
\frac12 c_\alpha
\right).
\label{eq:OU_factorized}
\end{align}

Equation~(\ref{eq:OU_factorized}) is the exact similarity-transformed
Ornstein--Uhlenbeck operator in factorized form. Expanding the product
gives

\begin{align}
M^{-1/2}
\mathcal{L}_{OU}
M^{1/2}
&=
\theta
\frac{\partial^2}{\partial v_\alpha^2}
+
\frac12
\frac{\partial c_\alpha}{\partial v_\alpha}
+
\frac12
c_\alpha
\frac{\partial}{\partial v_\alpha}
-
\frac12
c_\alpha
\frac{\partial}{\partial v_\alpha}
-
\frac{c_\alpha c_\alpha}{4\theta}.
\end{align}
Since
\[
\frac{\partial c_\alpha}{\partial v_\beta}
=
\delta_{\alpha\beta},
\]
one has
\[
\frac{\partial c_\alpha}{\partial v_\alpha}
=
d,
\]
and the two mixed terms cancel identically, yielding
\begin{equation}
M^{-1/2}
\mathcal{L}_{OU}
M^{1/2}
=
\theta
\frac{\partial^2}{\partial v_\alpha\partial v_\alpha}
+
\frac{d}{2}
-
\frac{c_\alpha c_\alpha}{4\theta}.
\label{eq:OU_final}
\end{equation}
Thus, after the similarity transformation, the drift term disappears
completely. The transformed operator consists only of the velocity-space
Laplacian, a constant contribution, and a quadratic multiplication
operator. The square-root choice is computationally distinguished because it eliminates the first-order velocity derivative.  This classical simplification is used here only to expose the subsequent Fock grading.  

\subsection{The Transformed LFH Operator}

We now collect the results obtained in the preceding subsections. The
streaming operator \eqref{eq:streaming_similarity_final} and the  Ornstein--Uhlenbeck operator \eqref{eq:OU_final} have each been
transformed separately under the square-root Maxwellian similarity
transformation. Combining the results yields the complete transformed
LFH operator ${\mathcal K}$,
\begin{equation}\label{eq:LFH_transformed_op}
    {\mathcal K}=\partial_t + v_\alpha{\partial_\alpha} + (\Gamma_t + v_\alpha\Gamma_\alpha) -
\gamma \left( \theta \frac{\partial}{\partial{v_\alpha}}\frac{\partial}{\partial{v_\alpha}} + \frac{d}{2} - \frac{c_\alpha c_\alpha}{4\theta} \right).
\end{equation}


The transformed LFH operator \eqref{eq:LFH_transformed_op} features a remarkable structural property.
Let us introduce the following pair of operators,
\begin{align}
\mathcal A_\alpha&=
\sqrt{\theta}\, \frac{\partial}{\partial v_\alpha} + \frac{1}{2\sqrt{\theta}}\,c_\alpha,
\label{eq:A_annihilation}
\\
\mathcal A_\alpha^\dagger&=
-\sqrt{\theta}\,\frac{\partial}{\partial v_\alpha} + \frac{1}{2\sqrt{\theta}}\,c_\alpha.
\label{eq:A_creation}
\end{align}
Direct computation recasts the operator \eqref{eq:LFH_transformed_op} as
\begin{equation}\label{eq:LFH_transformed_op_repr}
    {\mathcal K}={\partial_t} + v_\alpha{\partial_\alpha}
    + (\Gamma_t + v_\alpha\Gamma_\alpha)  +
\gamma \mathcal A^\dagger_\alpha \mathcal A_\alpha
\end{equation}
%
We also note that operators \eqref{eq:A_annihilation} and \eqref{eq:A_creation} satisfy commutation relations of Weyl algebra,
\begin{equation}
\begin{aligned}
{}[\mathcal A_\alpha,\mathcal A_\beta]
&=0,
\\
[\mathcal A_\alpha^\dagger,\mathcal A_\beta^\dagger]
&=0,
\\
[\mathcal A_\alpha,\mathcal A_\beta^\dagger]
&=
\delta_{\alpha\beta}.
\label{eq:A_ladderCCR}
\end{aligned}
\end{equation}
Equation~\eqref{eq:LFH_transformed_op_repr} is a velocity-coordinate
realization of the Weyl algebra.  The purpose of the next sections is to
separate this coordinate realization from the underlying algebraic state and
to pull back the complete parameterized kinetic dynamics to an abstract Fock
space.  The local velocity operators are intertwined algebraically, whereas
the propagation sector requires additional care because the coordinate
realization itself depends on the local fields.  We therefore first introduce
the abstract Fock representation and admissible coordinate realizations, and
then derive the pull-back of both local operators and external derivatives.
In this formulation the Fock ket depends only on $(\bm x,t)$; velocity
dependence is carried by the chosen coordinate realization.

For later use, we split operator \eqref{eq:LFH_transformed_op_repr} as follows,
\begin{equation}
    \label{eq:LFH_split}
    \mathcal K=\mathcal K_{\rm tr}+\mathcal K_{\rm con}+\mathcal{K}_{\rm OU},
\end{equation}
where 
\begin{align}
    \mathcal K_{\rm tr}&={\partial_t} + v_\alpha{\partial_\alpha},\label{eq:Ktr}\\
    \mathcal K_{\rm con}&=\Gamma_t + v_\alpha\Gamma_\alpha,\label{eq:Kcon}\\
    \mathcal{K}_{\rm OU}&=
\gamma \mathcal A^\dagger_\alpha \mathcal A_\alpha.\label{eq:KOU}
\end{align}
We refer to $\mathcal K_{\rm tr}$ as the space--time propagation sector.  The operator $\mathcal K_{\rm con}$ contains the local-Maxwellian connection generated by the square-root transformation, while $\mathcal K_{\rm OU}$ is the transformed Ornstein--Uhlenbeck relaxation sector.

\section{The Fock Representation of the Kinetic Space}\label{eq:Fock_space}

\subsection{The Canonical Operator Algebra and Fock Space}\label{sec:fock_space}
This subsection records the small amount of Fock-space machinery needed below. 
The Fock space $\mathscr H$ is understood here as a complex Hilbert space, as is
standard for the canonical operator algebra and the associated adjoint
operation. 
We use Dirac notation only as compact Hilbert-space notation.  A vector of $\mathscr H$ is written as a \emph{ket} $|\chi\rangle$; its Hilbert-space dual is the \emph{bra} $\langle\chi|$, and $\langle\chi|\Psi\rangle$ is the Fock inner product.  A one-input linear map $\widehat{\mathcal O}:\mathscr H\to\mathscr H$ will be called an \emph{operator} and acts as $\widehat{\mathcal O}|\chi\rangle$.  No quantum-mechanical interpretation is assumed.
\begin{remark}
The LFH kinetic equation considered in this work has real-valued
coefficients and preserves the real subspace of the representation.
Accordingly, the physical kinetic states may be restricted to the real
sector of the Fock space whenever no complexification is required.
\end{remark}

We introduce the canonical operators 
$\hat a_\alpha$
and
$\hat a_\alpha^\dagger$ acting on a Fock space $\mathscr H$ and 
satisfying the commutation relations
\begin{equation}
\begin{aligned}
{}[\hat a_\alpha, \hat a_\beta]
&=0,
\\
[\hat a_\alpha^\dagger, \hat a_\beta^\dagger]
&=0,
\\
[\hat a_\alpha, \hat a_\beta^\dagger]
&=
\delta_{\alpha\beta}.
\label{eq:ladderCCR}
\end{aligned}
\end{equation}
The construction of the Fock space begins with a distinguished vacuum state $|0\rangle$ satisfying
\begin{equation}
\hat a_\alpha|0\rangle=0,
\qquad
\alpha=1,\ldots,d.
\label{eq:vacuum}
\end{equation}
The simultaneous kernel of all annihilation operators $\hat{a}_\alpha$ is one-dimensional.
Choosing its normalized vector defines the vacuum state. The Fock basis is
then generated by repeated application of the operators
$\hat a_\alpha^\dagger$.
For the multi-index
\[
\bm n=(n_1,\ldots,n_d), \qquad n_\alpha=0,1,2,\ldots,
\]
with
\begin{equation}
\bm n!:=\prod_{\alpha=1}^{d}n_\alpha!,
\qquad
|\bm n|:=\sum_{\alpha=1}^{d}n_\alpha.
\label{eq:multiindex_conventions}
\end{equation}
the normalized Fock basis vectors are defined by
\begin{equation}
|\bm n\rangle
=
\prod_{\alpha=1}^{d}
\frac{1}
{\sqrt{n_\alpha!}}\left(\hat a_\alpha^\dagger\right)^{n_\alpha}
|0\rangle .
\label{eq:number_states}
\end{equation}
The vectors
$|\bm n\rangle$
form a complete orthonormal basis,
\begin{align}
\langle\bm m|\bm n\rangle
&=
\delta_{\bm m,\bm n},
\label{eq:orthogonality}
\\
\sum_{\bm n}
|\bm n\rangle
\langle\bm n|
&=
\widehat{I} .
\label{eq:completeness}
\end{align}
The canonical generators act algebraically on this basis,
\begin{align}
\hat a_\alpha
|\bm n\rangle
&=
\sqrt{n_\alpha}\,
|\bm n-\bm e_\alpha\rangle ,
\label{eq:a_action}
\\
\hat a_\alpha^\dagger
|\bm n\rangle
&=
\sqrt{n_\alpha+1}\,
|\bm n+\bm e_\alpha\rangle ,
\label{eq:adag_action}
\end{align}
where
$\bm e_\alpha$
denotes the unit multi-index in the
$\alpha$-th direction.
Consequently, the number operator
\begin{equation}
\hat{\mathcal N}
=
\sum_{\alpha=1}^{d}\hat a_\alpha^\dagger \hat a_\alpha
\label{eq:number_operator}
\end{equation}
acts diagonally,
\begin{equation}
\hat{ \mathcal N}|\bm n\rangle
=|\bm n|
|\bm n\rangle,
\label{eq:number_operator_action}
\end{equation}
where $|\bm n|$ is defined in Eq.~\eqref{eq:multiindex_conventions}.
Any state vector can be uniquely represented in the Fock basis \eqref{eq:number_states}. 
Our objective is to represent pertinent kinetic states as 
\begin{equation}
|\Psi(\bm x,t)\rangle = \sum_{\bm n} \psi_{\bm n}(\bm x,t) |\bm n\rangle.
\label{eq:Fock_expansion}
\end{equation}
The coefficients $\psi_{\bm n}(\bm x,t)$ of the Fock-space representation \eqref{eq:Fock_expansion} depend only on space and time, 
while the vectors $|\bm n\rangle$ span the internal Fock-space degrees of freedom
generated by the canonical operator algebra.
%
No coordinate representation of the Fock space has yet been introduced.
In particular, the relation between the Fock-space state
$|\Psi(\bm x,t)\rangle$ and a velocity-dependent coordinate state, such as
the square-root-transformed field $\Phi(\bm v,\bm x,t)$ of
Eq.~\eqref{eq:similarity}, remains to be established.  We now introduce this
relation through a local realization map and subsequently intertwine the
kinetic operators with their Fock-space representatives.

\subsection{Admissible Coordinate Realizations of Fock States}
\label{sec:admissible_realizations}

An admissible coordinate realization is an invertible linear map
\begin{equation}
\mathcal R:\mathscr H\longrightarrow\mathscr F
\label{eq:realization_map}
\end{equation}
from the abstract Fock space to a chosen local function space.  Its coordinate
image of an abstract state is denoted $\mathcal R|\Psi\rangle$.  At this point this is only a
change of representation of states; the corresponding operator intertwining
will be formulated in Sec.~\ref{sec:intertwining_general}.

For a realization of the canonical algebra we keep track of the distinguished
vacuum and generators,
\begin{equation}
\mathcal V_{\mathcal R}:=\mathcal R|0\rangle,
\qquad
\mathscr a^{(\mathcal R)}_\alpha
:=\mathcal R\hat a_\alpha\mathcal R^{-1},
\qquad
\mathscr a^{(\mathcal R)\dagger}_\alpha
:=\mathcal R\hat a_\alpha^\dagger\mathcal R^{-1}.
\label{eq:maxwell_realization_vacuum_generators}
\end{equation}
The vacuum relation is inherited automatically,
\begin{equation}
\mathscr a^{(\mathcal R)}_\alpha\mathcal V_{\mathcal R}=0,
\end{equation}
and the coordinate basis is generated by
\begin{equation}
\mathcal R|\bm n\rangle
=
\prod_{\alpha=1}^{d}
\frac{\bigl(\mathscr a^{(\mathcal R)\dagger}_\alpha\bigr)^{n_\alpha}}
{\sqrt{n_\alpha!}}\,\mathcal V_{\mathcal R}.
\label{eq:maxwell_realization_basis_from_vacuum}
\end{equation}

The connection between an abstract Fock state and a velocity-coordinate
state is therefore
\begin{equation}
\Phi_{\mathcal R}(\bm v,\bm x,t)
=\mathcal R(\bm x,t)|\Psi(\bm x,t)\rangle
=\sum_{\bm n}\psi_{\bm n}(\bm x,t)\,\mathcal R|\bm n\rangle.
\label{eq:Fock_expansion_coordinate}
\end{equation}
The coordinate representative $\Phi_{\mathcal R}$ should not be
identified with the physical distribution $f$, nor with the particular
square-root state $\Phi$ of Eq.~\eqref{eq:similarity} for an arbitrary
choice of $\mathcal R$.  The Hermite-function realization introduced below
will be used as the reference realization for which the physical
reconstruction is already known.  Any admissible realization related to it
by an invertible local similarity transformation inherits its reconstruction
map by transforming back to that reference coordinate state.  The resulting
coordinate-independent map from the abstract ket to $f$ is made explicit in
Sec.~\ref{sec:coordinate_realizations_equivalence}; the broader covariance
statement is deferred to Sec.~\ref{sec:covariance}.

Different admissible maps may give different coordinate functions for the
same abstract ket.  Two examples needed below are the Hermite-polynomial and
Hermite-function realizations.

\subsubsection{Hermite Polynomial Realization}
\label{sec:Hermite_poly}

The coordinate images of the Fock basis can now be identified explicitly.
We first consider the coordinate realization associated with the vacuum
pairing
\begin{equation}
    \mathcal R_P|0\rangle = 1.\label{eq:P_vacuum}
\end{equation}
Introducing the dimensionless peculiar velocity
\begin{equation}
    \xi_\alpha
    =
    \frac{c_\alpha}{\sqrt{\theta}},
\end{equation}
the coordinate realization of the canonical operators becomes
\begin{align}
    \mathscr a^{(P)}_\alpha
    &=
    \frac{\partial}{\partial \xi_\alpha},
    \\
    \mathscr a_\alpha^{(P)\dagger}
    &=
    -\frac{\partial}{\partial \xi_\alpha}
    +
    \xi_\alpha.
\end{align}
Accordingly, the coordinate image of the Fock basis is
\begin{equation}
    \phi_{\bm n}^{(P)}(\bm\xi)
    :=
    \mathcal R_P|\bm n\rangle
    =
    \frac{1}{\sqrt{\bm n!}}
    \prod_{\alpha=1}^{d}
    \left(
    \xi_\alpha
    -
    \frac{\partial}{\partial \xi_\alpha}
    \right)^{n_\alpha}
    1.
    \label{eq:coordinate_fock_basis_hermite}
\end{equation}
We use throughout the physicists' Hermite polynomials,
\begin{equation}
    H_n(z)
    =
    (-1)^n e^{z^2}
    \frac{d^n}{dz^n}
    e^{-z^2}.
    \label{eq:hermite_rodrigues}
\end{equation}
For the ladder normalization used above, the corresponding creation
identity is
\begin{equation}
    \left(
    \xi-\frac{d}{d\xi}
    \right)^n1
    =
    2^{-n/2}
    H_n\left(\frac{\xi}{\sqrt2}\right).
    \label{eq:hermite_creation_identity}
\end{equation}
Consequently,
\begin{equation}
    \phi_{\bm n}^{(P)}(\bm\xi)
    =
    \frac{
    H_{\bm n}\!\left(\bm\xi/\sqrt2\right)
    }{
    \sqrt{2^{|\bm n|}\bm n!}
    } ,
    \label{eq:coordinate_fock_hermite}
\end{equation}
where
\begin{equation}
    H_{\bm n}(\bm z)
    =
    \prod_{\alpha=1}^{d}
    H_{n_\alpha}(z_\alpha).
\end{equation}

The natural weight is the normalized Maxwellian $W$ of
Eq.~\eqref{eq:W_maxwell}.  With the normalization in
Eq.~\eqref{eq:coordinate_fock_hermite},
\begin{equation}
    \int_{\mathbb R^d}
    W(\bm v;\bm u,\theta)\,
    \phi_{\bm m}^{(P)}(\bm\xi)
    \phi_{\bm n}^{(P)}(\bm\xi)
    \,d^dv
    =
    \delta_{\bm m,\bm n}.
    \label{eq:polynomial_basis_orthogonality}
\end{equation}
Thus $\mathcal R_P$ is naturally viewed as a realization in
$L^2(W\,d^dv)$.

\subsubsection{Hermite Function Realization}
\label{sec:Hermite_function}

The Hermite function realization is obtained from the Hermite polynomial
realization by multiplication with the square root of the normalized
Maxwellian weight. Define
\begin{equation}
    \mathcal S_W:\Phi(\bm v)
    \mapsto
    \sqrt{W(\bm v,\bm x,t)}\,\Phi(\bm v),
    \qquad
    \mathcal R_H
    =
    \mathcal S_W\mathcal R_P.
    \label{eq:Hermite_similarity_map}
\end{equation}
With Eq.~\eqref{eq:sqrtM}, one has
\begin{equation}
    \sqrt{W}
    =
    \frac{1}{(2\pi\theta)^{d/4}}
    \exp\left(
    -\frac{\xi_\alpha\xi_\alpha}{4}
    \right).
\end{equation}
The similarity-transformed ladder operators are
\begin{align}
    \mathscr a^{(H)}_\alpha
    &=
    \mathcal S_W
    \mathscr a^{(P)}_\alpha
    \mathcal S_W^{-1}
    =
    \frac{\partial}{\partial \xi_\alpha}
    +
    \frac{\xi_\alpha}{2},\label{eq:a_H_annihilation}
    \\
    \mathscr a_\alpha^{(H)\dagger}
    &=
    \mathcal S_W
    \mathscr a_\alpha^{(P)\dagger}
    \mathcal S_W^{-1}
    =
    -\frac{\partial}{\partial \xi_\alpha}
    +
    \frac{\xi_\alpha}{2}.\label{eq:a_H_creation}
\end{align}
We note that the ladder operators in the Hermite functions coordinate realization coincide with the previously introduced operators \eqref{eq:A_annihilation} and \eqref{eq:A_creation}, and that the  annihilation operator \eqref{eq:a_H_annihilation} is adapted to the square root Gaussian vacuum:
\begin{equation}
    \mathscr a^{(H)}_\alpha\sqrt W=0.
\end{equation}
The corresponding coordinate images of the Fock basis are
\begin{equation}
    \varphi_{\bm n}^{(H)}
    :=
    \mathcal R_H|\bm n\rangle
    =
    \frac{1}{\sqrt{\bm n!}}
    \prod_{\alpha=1}^{d}
    \left(
    \mathscr a_\alpha^{(H)\dagger}
    \right)^{n_\alpha}
    \sqrt W.
    \label{eq:hermite_function_basis}
\end{equation}
Using
\begin{equation}
    \left(
    -\frac{d}{d\xi}
    +
    \frac{\xi}{2}
    \right)^n
    e^{-\xi^2/4}
    =
    2^{-n/2}
    H_n\left(\frac{\xi}{\sqrt2}\right)
    e^{-\xi^2/4},
    \label{eq:scaled_hermite_identity}
\end{equation}
we obtain
\begin{equation}
    \varphi_{\bm n}^{(H)}(\bm\xi)
    =
    \sqrt W\,
    \frac{
    H_{\bm n}\!\left(\bm\xi/\sqrt2\right)
    }{
    \sqrt{2^{|\bm n|}\bm n!}
    }
    =
    \sqrt W\,
    \phi_{\bm n}^{(P)}(\bm\xi).
    \label{eq:hermite_function_explicit}
\end{equation}

The Hermite-function realization is orthonormal in the ordinary
Lebesgue space $L^2(d^dv)$:
\begin{equation}
    \int_{\mathbb R^d}
    \varphi_{\bm m}^{(H)}(\bm v)
    \varphi_{\bm n}^{(H)}(\bm v)
    \,d^dv
    =
    \delta_{\bm m,\bm n}.
    \label{eq:hermite_function_orthogonality}
\end{equation}

The two coordinate realizations of the same abstract Fock basis can
therefore be summarized as
\begin{equation}
\begin{aligned}
\mathcal R_P|\bm n\rangle
&=
\phi_{\bm n}^{(P)}(\bm\xi)
=
\frac{
H_{\bm n}\!\left(\bm\xi/\sqrt2\right)
}{
\sqrt{2^{|\bm n|}\bm n!}
},
&
\mathcal R_P|0\rangle
&=
1,
\\[2mm]
\mathcal R_H|\bm n\rangle
&=
\varphi_{\bm n}^{(H)}(\bm\xi)
=
\sqrt W\,
\frac{
H_{\bm n}\!\left(\bm\xi/\sqrt2\right)
}{
\sqrt{2^{|\bm n|}\bm n!}
},
&
\mathcal R_H|0\rangle
&=
\sqrt W .
\end{aligned}
\label{eq:two_fock_coordinate_realizations}
\end{equation}

\subsubsection{Coordinate Realizations and Representation Equivalence}
\label{sec:coordinate_realizations_equivalence}

The two bases above are coordinate images of the same abstract Fock basis and
are related by
\begin{equation}
\mathcal R_H=\mathcal S_W\mathcal R_P,
\qquad
\mathcal S_W\Phi=\sqrt W\,\Phi.
\label{eq:coordinate_realization_similarity}
\end{equation}
Consequently their coordinate ladder operators are related by conjugation and
both realize the same canonical commutation relations.  This elementary
example already separates the abstract Fock state from its coordinates: the
vacuum may be represented by $1$ or by $\sqrt W$ without changing the
underlying ket $|0\rangle$.

The same distinction is required when reconstructing the physical
distribution.  In the Hermite-function realization the coordinate state is
the square-root-transformed state of Eq.~\eqref{eq:similarity}, so that
\begin{equation}
f=\sqrt n\,\sqrt W\,\Phi_H,
\qquad
\Phi_H=\mathcal R_H|\Psi\rangle.
\label{eq:physical_reconstruction_H}
\end{equation}
For any admissible realization $\mathcal R$ in the same similarity class,
$\Phi_{\mathcal R}=\mathcal R|\Psi\rangle$ is first transformed back to
$\Phi_H$.  Thus
\begin{equation}
\mathcal B_{\mathcal R}
:=\sqrt n\,\sqrt W\,\mathcal R_H\mathcal R^{-1},
\qquad
f=\mathcal B_{\mathcal R}\Phi_{\mathcal R}.
\label{eq:physical_reconstruction_general}
\end{equation}
The composite map from the abstract Fock state to the physical distribution,
\begin{equation}
\mathcal Q
:=\mathcal B_{\mathcal R}\mathcal R
=\sqrt n\,\sqrt W\,\mathcal R_H,
\qquad
f=\mathcal Q|\Psi\rangle,
\label{eq:physical_realization_Q}
\end{equation}
is therefore independent of the coordinate representative within this
similarity class.  No separate reconstruction rule has to be guessed for each
realization.  The dependence of the factors on the chosen coordinates, and
the invariance of their product, will be revisited in
Sec.~\ref{sec:covariance}.

\section{Fock-Space Formulation of the LFH Operator}\label{sec:intertwining}
\subsection{General Setting for Operator Intertwining}
\label{sec:intertwining_general}

\subsubsection{Local operators}

Let $\mathcal R(\lambda)$ be an admissible invertible realization on a common
dense algebraic core, where $\lambda$ denotes the local realization
parameters.  A coordinate operator $\mathcal O$ and a Fock-space operator
$\widehat{\mathcal O}$ are \emph{intertwined} by $\mathcal R$ when
\begin{equation}
\mathcal R\widehat{\mathcal O}=\mathcal O\mathcal R,
\qquad
\widehat{\mathcal O}=\mathcal R^{-1}\mathcal O\mathcal R.
\label{eq:general_intertwining}
\end{equation}
For the canonical generators,
\begin{equation}
\mathcal R\hat a_\alpha
=\mathscr a_\alpha^{(\mathcal R)}\mathcal R,
\qquad
\mathcal R\hat a_\alpha^\dagger
=\mathscr a_\alpha^{(\mathcal R)\dagger}\mathcal R.
\label{eq:ladder_intertwining}
\end{equation}
Because the ladder operators generate the Weyl algebra, every polynomial
$P$ obeys
\begin{equation}
\mathcal R P(\hat a,\hat a^\dagger)
=
P(\mathscr a^{(\mathcal R)},\mathscr a^{(\mathcal R)\dagger})\mathcal R.
\label{eq:algebraic_intertwining}
\end{equation}
Coefficients of $P$ may depend on $\lambda$ and its external derivatives;
this does not affect the local pull-back because those coefficients act by
multiplication with respect to the velocity coordinate.

\subsubsection{External derivatives and the differential connection}
\label{sec:diff_connection}

The extension to propagation is less immediate.  Let
$\partial_\mu$, $\mu\in\{t,1,\ldots,d\}$, denote an external derivative and
assume that $\mathcal R=\mathcal R(\lambda(\bm x,t))$.  The Leibniz rule gives
\begin{align}
\partial_\mu\bigl(\mathcal R|\Psi\rangle\bigr)
&=(\partial_\mu\mathcal R)|\Psi\rangle
+\mathcal R\partial_\mu|\Psi\rangle.
\label{eq:derivative_realization_product}
\end{align}
Hence the pull-back of the external derivative is the composition operator
\begin{equation}
\widehat{\partial}_\mu
:=\mathcal R^{-1}\circ\partial_\mu\circ\mathcal R
=\partial_\mu+\widehat{\mathcal G}_\mu,
\qquad
\widehat{\mathcal G}_\mu
:=\mathcal R^{-1}(\partial_\mu\mathcal R).
\label{eq:external_diff_pullback}
\end{equation}
We keep the global parameter notation
$\lambda=(n,\bm u,\theta)$.  For the polynomial and Hermite-function
realizations used below, $\mathcal R$ is independent of $n$, so the chain
rule reduces to
\begin{equation}
\widehat{\mathcal G}_\mu
=
\mathcal R^{-1}
\left(
\frac{\partial\mathcal R}{\partial u_\beta}\partial_\mu u_\beta
+
\frac{\partial\mathcal R}{\partial\theta}\partial_\mu\theta
\right).
\label{eq:diff_connector}
\end{equation}
We call $\widehat{\mathcal G}_\mu$ the \emph{differential connection}.  It
should be distinguished from the local-Maxwellian connection $\Gamma_\mu$ of
Sec.~\ref{sec:LFH_transform_similarity}.  The latter belongs to the
square-root-transformed kinetic generator; the former is induced by the
space--time dependence of the chosen coordinate realization.  They have
different origins even though both contain derivatives of the local
parameters.

The construction is explicit for an arbitrary realization.  If
\begin{equation}
\varphi_{\bm n}(\bm v;\lambda)=\mathcal R(\lambda)|\bm n\rangle,
\label{eq:coordinate_basis_expanded_general}
\end{equation}
then differentiation and re-expansion in the same basis gives
\begin{equation}
\frac{\partial\varphi_{\bm n}}{\partial\lambda}
=
\sum_{\bm m}\varphi_{\bm m}
\left\langle\bm m\left|
\mathcal R^{-1}\frac{\partial\mathcal R}{\partial\lambda}
\right|\bm n\right\rangle.
\label{eq:diff_connection_general}
\end{equation}
Here $\partial/\partial\lambda$ is used compactly to denote
differentiation with respect to any one of the local realization parameters;
no summation over a parameter label is implied.  Thus the connection is not a formal symbol: its matrix elements are exactly
the coefficients obtained by differentiating the coordinate basis.

Finally, let $P$ be a local polynomial in the coordinate Weyl generators.
Preserving operator ordering, the pull-back of a first-order quasilinear term
is
\begin{equation}
\mathcal R^{-1}
\left[P(\mathscr a^{(\mathcal R)},\mathscr a^{(\mathcal R)\dagger})
\partial_\mu\right]\mathcal R
=
P(\hat a,\hat a^\dagger)
\left[\partial_\mu+\mathcal R^{-1}(\partial_\mu\mathcal R)\right].
\label{eq:D_intertwining}
\end{equation}
Equation~\eqref{eq:D_intertwining} is the extension of algebraic
intertwining needed for kinetic propagation.

\subsection{Fock-Space Representation of the Local Operator Sector}
\label{sec:local_operator_sectors}

We now begin the sector-by-sector construction of the Fock-space representation of the
transformed LFH operator. In the following, we choose to work with the Hermite-function coordinate realization of Sec.~\ref{sec:Hermite_function}. The local sector will be considered first.

\subsubsection{Local Maxwellian Connection Sector}
\label{sec:connector_intertwining}
For the pull-back of the local Maxwellian connection, it is sufficient to note that the multiplication by the peculiar velocity $\bm c$ is represented by the operators \eqref{eq:a_H_annihilation} and \eqref{eq:a_H_creation} as 
\begin{equation}
     c_\alpha = \sqrt\theta\left(\mathscr a^{(H)}_\alpha+\mathscr a_\alpha^{(H)\dagger}\right).
\end{equation}
Consequently, the corresponding intertwined Fock-space operator is
\begin{equation}\label{eq:c_alpha_Fock}
    \widehat c_\alpha = \sqrt \theta\left(\hat a_\alpha+\hat a_\alpha^\dagger\right).
\end{equation}
Furthermore, upon inspection of the functions $\Gamma_\mu$ \eqref{eq:Gamma_xt_explicit}, the local Maxwellian connection \eqref{eq:Kcon} is a polynomial in the peculiar velocity, with the coefficients depending on the parameters $n$, $\bm u$ and $\theta$, as well as on their time and space derivatives. Thus, for the purely
multiplicative polynomial connection terms, the
substitution of the coordinate velocity variables by their Fock-space
counterparts is unambiguous.
Hence, we can write
\begin{equation}\label{eq:GG_Fock}
\begin{aligned}
 &   \widehat\Gamma_\mu=\Gamma_\mu\left(\widehat{\bm c}; \partial_\mu n,\partial_\mu\bm u,\partial_\mu\theta,n,\theta\right),
\end{aligned}
\end{equation}
and the intertwining relation for the local Maxwellian connection becomes
\begin{align}
\widehat{\mathcal K}_{\rm con}
=
\left(
\widehat\Gamma_t
+
\widehat v_\alpha
\widehat\Gamma_\alpha
\right)
=
\mathcal R_H^{-1}\left(
\Gamma_t
+
v_\alpha\Gamma_\alpha
\right)
\mathcal R_H
=
\mathcal R_H^{-1}\mathcal K_{\rm con}\mathcal R_H,
\label{eq:Kcon_intertwining}
\end{align}
where 
\begin{equation}\label{eq:v_alpha_Fock}
    \widehat{v}_\alpha=u_\alpha I+\widehat c_\alpha,
\end{equation}
is the Fock-space representation of the velocity.
Thus, the connector sector is lifted to the Fock space by replacing the
complete velocity-polynomial structure of the coordinate operator by
the corresponding polynomial in the Fock-space velocity operators.

\subsubsection{Ornstein--Uhlenbeck Sector}
\label{sec:OU_intertwining}

For the Ornstein--Uhlenbeck contribution \eqref{eq:KOU}, the corresponding Fock-space operator is defined by
\begin{equation}
\widehat{\mathcal K}_{\rm OU}
=
\gamma
\hat a_\alpha^\dagger
\hat a_\alpha .
\label{eq:Khat_OU}
\end{equation}
Using the intertwining relations for the two canonical generators, we
find
\begin{align}
\mathcal R_H
\widehat{\mathcal K}_{\rm OU}
&=
\gamma
\mathcal R_H
\hat a_\alpha^\dagger \mathcal R_H^{-1}\mathcal R_H
\hat a_\alpha \mathcal R_H^{-1}\mathcal R_H
=
\gamma
\mathscr a_\alpha^{(H)\dagger}
\mathscr a^{(H)}_\alpha
\mathcal R_H
=
\mathcal K_{\rm OU}
\mathcal R_H .
\label{eq:KOU_intertwining}
\end{align}
Thus the Ornstein--Uhlenbeck operator is represented on the Fock space
by the number operator \eqref{eq:number_operator};
In the Fock basis, its action is diagonal \eqref{eq:number_operator}.
This completes the algebraic construction of the local sectors. The
remaining part of the transformed LFH operator is the propagation
sector. Unlike the local terms considered above, it contains
space-time derivatives acting on a coordinate realization that depends
on the local macroscopic fields. Its Fock-space representation
therefore requires an explicit treatment of the variation of the
coordinate basis.


\subsection{Fock-space Representation of the Transport Sector}
\subsubsection{Pull-back and differential connection}

The transport operator $\mathcal K_{\rm tr}$  \eqref{eq:Ktr} 
belongs to the class of quasilinear differential operators considered in Sec.~\ref{sec:intertwining_general}. Its Fock-space pull-back is therefore obtained directly from Eq.~\eqref{eq:D_intertwining}. Defining
\[
\widehat{\mathcal K}_{\rm tr}
:=\mathcal R^{-1}\mathcal K_{\rm tr}\mathcal R,
\]
Eq.~\eqref{eq:D_intertwining} gives
\begin{equation}
\widehat{\mathcal K}_{\rm tr}
=
\partial_t+\widehat{\mathcal G}_t
+
\widehat v_\alpha
\left(\partial_\alpha+\widehat{\mathcal G}_\alpha\right).
\label{eq:Ktr_general_pullback}
\end{equation}
This is the transport pull-back for an arbitrary admissible realization.
For the Hermite-function realization used in the explicit LFH calculation we
write
\begin{equation}
\widehat{\mathcal C}_\mu
:=\widehat{\mathcal G}^{(H)}_\mu
=\mathcal R_H^{-1}(\partial_\mu\mathcal R_H),
\label{eq:C_time}
\end{equation}
so that
\begin{equation}
\widehat{\mathcal K}_{\rm tr}^{(H)}
=
\partial_t+\widehat{\mathcal C}_t
+
\widehat v_\alpha
\left(\partial_\alpha+\widehat{\mathcal C}_\alpha\right).
\label{eq:Ktr_connector_form}
\end{equation}
It remains only to evaluate these Hermite-function connectors. 
According to the general construction of Sec.~\ref{sec:intertwining_general}, their matrix elements are determined by differentiating the coordinate basis functions with respect to the local parameters $\bm{u}$ and $\theta$, and re-expanding the resulting functions in the same Hermite basis.

\subsubsection{Evaluation of the Differential Connections}
\label{sec:transport_connector_evaluation}

The Hermite-function realization map $\mathcal R_H$ depends on the external variables $(\bm x,t)$ through the local parameters $u_\beta(\bm x,t)$ and $\theta(\bm x,t)$. Consequently, its external derivatives may be evaluated by the chain rule:
\begin{equation}
\partial_\mu\mathcal R_H
=
(\partial_\mu u_\beta)
\frac{\partial\mathcal R_H}{\partial u_\beta}
+
(\partial_\mu\theta)
\frac{\partial \mathcal R_H}{\partial \theta},
\qquad
\mu\in\{t,1,\ldots,d\}.
\label{eq:R_chain_rule}
\end{equation}
It is therefore sufficient to determine the two elementary parameter connectors
\begin{align}
\widehat{\mathcal C}_{u_\beta}
&:=
\mathcal R_H^{-1}
\left(
\frac{\partial\mathcal R_H}{\partial u_\beta}
\right),
\label{eq:C_u_elementary}\\
\widehat{\mathcal C}_{\theta}
&:=
\mathcal R_H^{-1}
\left(
\frac{\partial\mathcal R_H}{\partial \theta}
\right).
\label{eq:C_theta_elementary}
\end{align}
The differential connection admits a simple geometric interpretation. Since the Hermite-function basis depends parametrically on the local Maxwellian through the macroscopic fields $\bm u$ and $\theta$, changes of these fields modify not only the expansion coefficients but also the basis itself. The differential connection describes how the coordinate realization adapted to the local Maxwellian changes.  The elementary connectors quantify the infinitesimal change of the local Hermite-function realization associated with variations of the parameters defining the reference Maxwellian. The velocity connection corresponds to translations of the Maxwellian center in velocity space, while the temperature connection corresponds to dilation associated with changes of its width.

The external differential connectors introduced in
Eq.~\eqref{eq:C_time} are then given by
\begin{equation}
\widehat{\mathcal C}_\mu
=
(\partial_\mu u_\beta)
\widehat{\mathcal C}_{u_\beta}
+
(\partial_\mu\theta)
\widehat{\mathcal C}_{\theta}.
\label{eq:C_mu_chain_rule}
\end{equation}
To evaluate the elementary connectors, we specialize the coordinate basis \eqref{eq:coordinate_basis_expanded_general} to the 
multidimensional Hermite-function basis \eqref{eq:hermite_function_explicit} associated with the local parameters $\bm u$ and $\theta$. According to the general result of Sec.~\ref{sec:diff_connection}, the matrix elements of the elementary connectors are determined by the expansions (see Eq.~\eqref{eq:diff_connection_general}),
\begin{align}
\frac{\partial\varphi^{(H)}_{\bm n}}{\partial u_\beta}
&=
\sum_{\bm m}
\varphi^{(H)}_{\bm m}
\left\langle
\bm m
\left|
\widehat{\mathcal C}_{u_\beta}
\right|
\bm n
\right\rangle,
\label{eq:du_basis_expansion}\\
\frac{\partial \varphi^{(H)}_{\bm n}}{\partial \theta}
&=
\sum_{\bm m}
\varphi^{(H)}_{\bm m}
\left\langle
\bm m
\left|
\widehat{\mathcal C}_{\theta}
\right|
\bm n
\right\rangle.
\label{eq:dtheta_basis_expansion}
\end{align}
Thus, the connector operators are obtained by differentiating the coordinate basis functions and re-expanding the result in the same basis.
For the Hermite-function realization adopted here, differentiation with respect to the local velocity component $u_\beta$ gives (details of the computation are given in Appendix \ref{app:hermite_differential_connectors}),
\begin{equation}
\frac{\partial \varphi^{(H)}_{\bm n}}{\partial u_\beta}
=
\frac{1}{2\sqrt{\theta}}
\left[
\sqrt{n_\beta+1}
\varphi^{(H)}_{\bm n+\bm e_\beta}
-
\sqrt{n_\beta}
\varphi^{(H)}_{\bm n-\bm e_\beta}
\right],
\label{eq:Hermite_du_identity}
\end{equation}
where $\bm e_\beta$ denotes the unit multi-index in the $\beta$th velocity direction. Comparing Eq.~\eqref{eq:Hermite_du_identity} with Eq.~\eqref{eq:du_basis_expansion}, one obtains
\begin{align}
\left\langle
\bm m
\left|
\widehat{\mathcal C}_{u_\beta}
\right|
\bm n
\right\rangle
=
\frac{1}{2\sqrt{\theta}}
\Big[
\sqrt{n_\beta+1}
\delta_{\bm m,\bm n+\bm e_\beta}
-
\sqrt{n_\beta}
\delta_{\bm m,\bm n-\bm e_\beta}
\Big].
\label{eq:C_u_matrix}
\end{align}
Using the standard Fock-space actions
\begin{align}
\hat a_\beta|\bm n\rangle
=
\sqrt{n_\beta}
|\bm n-\bm e_\beta\rangle,
\quad
\hat a_\beta^\dagger|\bm n\rangle
=
\sqrt{n_\beta+1}
|\bm n+\bm e_\beta\rangle,
\end{align}
the matrix elements \eqref{eq:C_u_matrix} are represented by
\begin{equation}
\widehat{\mathcal C}_{u_\beta}
=
\frac{1}{2\sqrt{\theta}}
\left(
\hat a_\beta^\dagger-\hat a_\beta
\right)
.
\label{eq:C_u_ladder}
\end{equation}
Equation~\eqref{eq:C_u_ladder} is therefore not introduced as an independent definition; it is the ladder-operator representation of the matrix obtained from the differentiated Hermite-functions basis.

Differentiation with respect to the local temperature parameter produces transitions by two excitation levels in each velocity direction:
\begin{align}
\frac{\partial \varphi^{(H)}_{\bm n}}{\partial \theta}
=
\frac{1}{4\theta}
\sum_\beta
\Big[
&
\sqrt{(n_\beta+1)(n_\beta+2)}
\varphi^{(H)}_{\bm n+2\bm e_\beta}
-
\sqrt{n_\beta(n_\beta-1)}
\varphi^{(H)}_{\bm n-2\bm e_\beta}
\Big].
\label{eq:Hermite_dtheta_identity}
\end{align}
Consequently,
\begin{align}
\left\langle
\bm m
\left|
\widehat{\mathcal C}_{\theta}
\right|
\bm n
\right\rangle
=
\frac{1}{4\theta}
\sum_\beta
\Big[
&
\sqrt{(n_\beta+1)(n_\beta+2)}
\delta_{\bm m,\bm n+2\bm e_\beta}
-
\sqrt{n_\beta(n_\beta-1)}
\delta_{\bm m,\bm n-2\bm e_\beta}
\Big].
\label{eq:C_theta_matrix}
\end{align}
Since
\begin{align}
\hat a_\beta^{\dagger 2}|\bm n\rangle
=
\sqrt{(n_\beta+1)(n_\beta+2)}
|\bm n+2\bm e_\beta\rangle,
\quad
\hat a_\beta^2|\bm n\rangle
=
\sqrt{n_\beta(n_\beta-1)}
|\bm n-2\bm e_\beta\rangle,
\end{align}
the matrix \eqref{eq:C_theta_matrix} admits the compact operator representation
\begin{equation}
\widehat{\mathcal C}_{\theta}
=
\frac{1}{4\theta}
\left(
\hat a_\beta^\dagger\hat a_\beta^\dagger
-
\hat a_\beta\hat a_\beta
\right)
,
\label{eq:C_theta_ladder}
\end{equation}
where summation over the repeated velocity component $\beta$ is understood.

Substituting Eqs.~\eqref{eq:C_u_ladder} and
\eqref{eq:C_theta_ladder} into the chain-rule relation
\eqref{eq:C_mu_chain_rule} gives
\begin{align}
\widehat{\mathcal C}_\mu
=
\frac{\partial_\mu u_\beta}{2\sqrt{\theta}}
\left(
\hat a_\beta^\dagger-\hat a_\beta
\right)
+
\frac{\partial_\mu\theta}{4\theta}
\left(
\hat a_\beta^\dagger\hat a_\beta^\dagger
-
\hat a_\beta\hat a_\beta
\right).
\label{eq:C_mu_explicit}
\end{align}

The differential connection \eqref{eq:C_mu_explicit}, together with
Eq.~\eqref{eq:Ktr_connector_form}, provide the explicit Fock-space representation of the transport sector for the Hermite-function realization.
Thus, the Fock-space representation of the transport operator is explicitly written as follows,
\begin{align}
\widehat{\mathcal K}_{\rm tr}
=
{\partial_t}
+
\widehat v_\alpha
{\partial_\alpha}
+
\frac{
\partial_tu_\beta
+
\widehat v_\alpha\partial_{\alpha}u_\beta
}{
2\sqrt{\theta}
}
\left(
\hat a_\beta^\dagger-\hat a_\beta
\right)
+
\frac{
\partial_t\theta
+
\widehat v_\alpha\partial_{\alpha}\theta
}{
4\theta
}
\left(
\hat a_\beta^\dagger \hat a_\beta^\dagger
-
\hat a_\beta \hat a_\beta
\right)
\label{eq:Ktr_hat_Hermite_expanded}
\end{align}
Together with the results of Sec.~\ref{sec:local_operator_sectors},  this concludes the construction of Fock-space representation of the LFH kinetic operator.


\subsection{Representation Theorem}

The preceding construction can be summarized without fixing a particular
coordinate basis.

\begin{theorem}[Intertwining representation of the LFH generator]
\label{thm:representation}
Let $\mathcal R(\lambda)$ be an admissible invertible local realization on a
common algebraic core, intertwining the canonical Weyl generators as in
Eq.~\eqref{eq:ladder_intertwining}.  Let $\mathcal K[\lambda]$ be a coordinate
kinetic generator whose velocity dependence is polynomial in those generators
and whose external dependence is first order.  Then
\begin{equation}
\widehat{\mathcal K}[\lambda]
:=\mathcal R^{-1}\mathcal K[\lambda]\mathcal R
\label{eq:representation_theorem_pullback}
\end{equation}
is well defined on the common core and satisfies
\begin{equation}
\mathcal R\widehat{\mathcal K}[\lambda]
=\mathcal K[\lambda]\mathcal R.
\label{eq:representation_theorem_intertwining}
\end{equation}
Every external derivative is pulled back according to
Eq.~\eqref{eq:external_diff_pullback}; no additional rule is required.
In particular, the theorem applies to the complete square-root-transformed LFH
generator.
\end{theorem}

\begin{proof}
Local polynomial factors follow from the Weyl-algebra intertwining
Eq.~\eqref{eq:algebraic_intertwining}.  Every first-order external derivative
follows from the Leibniz identity Eq.~\eqref{eq:external_diff_pullback}, and
Eq.~\eqref{eq:D_intertwining} preserves the ordering of quasilinear products.
The LFH generator is a sum of terms of these two types, which proves the
claim.
\end{proof}

For the Hermite-function realization, completeness and
Eq.~\eqref{eq:hermite_function_orthogonality} identify $\mathcal R_H$ with a
unitary map from the completed abstract Fock space to $L^2(\mathbb R^d,d^dv)$;
thus the abstract theorem has a concrete Hilbert-space realization.  We have
nevertheless stated the theorem for general $\mathcal R$ because the
representation structure, rather than the Hermite coordinates used to
calculate one representative, will be the object whose covariance is studied
in Sec.~\ref{sec:covariance}.

At this stage the fields $\lambda=(n,\bm u,\theta)$ remain parameters of a
family of intertwined operators and local realizations.  The theorem neither
identifies them with the physical moments of the represented state nor
determines their evolution.  That is the separate compatibility problem to
which we now turn.

\section{Compatibility of the Fock Representation}
\label{sec:fock_compatibility}

The operator representation is only one side of the kinetic theory.  Physical
moments are functionals of the distribution $f$ and must be represented in a
way that is independent of the chosen coordinate realization.  We first
construct these dual functionals, then derive a general transport--moment
identity, and only afterwards impose compatibility between the local
parameters and the moments of the represented state.

\subsection{Moment Functionals in Fock Space}
\label{sec:fock_moment_functionals}

For a polynomial velocity observable $m(\bm v)$, the corresponding physical
moment of the distribution is
\begin{equation}
\mathcal M_m[f]
:=\int_{\mathbb R^d}m(\bm v)f(\bm v)\,d^dv .
\label{eq:general_moment}
\end{equation}
The reconstruction map introduced in
Eq.~\eqref{eq:physical_realization_Q} gives $f=\mathcal Q|\Psi\rangle$;
therefore, at the realization-independent level,
\begin{equation}
\mathcal M_m[f]
=\mathcal M_m\!\left[\mathcal Q|\Psi\rangle\right].
\label{eq:moment_functional_fock}
\end{equation}
No additional Fock-space operator is associated with this statement: a moment
is a scalar linear functional of the physical distribution.

In the Hermite-function realization the same functional has a particularly
simple evaluation.  From
\begin{equation}
f=\sqrt n\sqrt W\,\mathcal R_H|\Psi\rangle
\label{eq:definition_f_RPsi}
\end{equation}
and orthonormality one has, for every ket $|\chi\rangle$,
\begin{equation}
\langle0|\chi\rangle
=\int_{\mathbb R^d}\sqrt W\,
(\mathcal R_H|\chi\rangle)(\bm v)\,d^dv.
\label{eq:vacuum_paring_H_functions}
\end{equation}
Multiplication by particle velocity is intertwined with
\begin{equation}
\widehat v_\alpha
= u_\alpha I+\sqrt\theta(\hat a_\alpha+\hat a_\alpha^\dagger),
\qquad
v_\alpha\mathcal R_H=\mathcal R_H\widehat v_\alpha,
\label{eq:moment_velocity_intertwining}
\end{equation}
and hence
\begin{equation}
m(\bm v)\mathcal R_H
=\mathcal R_Hm(\widehat{\bm v}).
\label{eq:moment_intertwining}
\end{equation}
It follows that
\begin{equation}
\mathcal M_m[f]
=\sqrt n\,\langle0|m(\widehat{\bm v})|\Psi\rangle.
\label{eq:moment_vacuum_pairing}
\end{equation}
For later use we denote the corresponding Hermite-function moment bra by
\begin{equation}
\langle M_m|
:=\langle0|m(\widehat{\bm v}),
\qquad
\mathcal M_m[f]=\sqrt n\,\langle M_m|\Psi\rangle.
\label{eq:moment_bra_general}
\end{equation}
Thus the bra is a convenient coordinate representative of the physical
moment functional; the physical definition remains Eq.~\eqref{eq:general_moment}.

\subsection{Transport--Moment Intertwining}
\label{sec:transport_moment_intertwining}

Let $\mathcal T=\partial_t+v_\alpha\partial_\alpha$ be the physical free
transport operator.  Let $\widehat{\mathcal K}_{\rm p}$ denote the complete
Fock-space propagation operator, including all terms generated by the local
parameter dependence of the chosen realization.  Since $\mathcal Q$ maps the
abstract ket directly to the physical distribution, propagation is
characterized by the physical intertwining relation
\begin{equation}
\mathcal Q\widehat{\mathcal K}_{\rm p}
=\mathcal T\mathcal Q.
\label{eq:physical_transport_intertwining}
\end{equation}
In the square-root Hermite-function realization,
\begin{equation}
\widehat{\mathcal K}_{\rm p}
=\widehat{\mathcal K}_{\rm tr}+\widehat{\mathcal K}_{\rm con},
\label{eq:propagation_intertwining}
\end{equation}
where the first term contains the differential connection and the second the
local-Maxwellian connection.  Their separate forms are not required in the
following result.

\begin{theorem}[Transport--moment intertwining]
\label{thm:transport_moment_intertwining}
Let $m=m(\bm v)$ have no explicit dependence on $(\bm x,t)$.  Then
\begin{equation}
\mathcal M_m\!\left[
\mathcal Q\widehat{\mathcal K}_{\rm p}|\Psi\rangle
\right]
=
\partial_t\mathcal M_m[f]
+\partial_\alpha\mathcal M_{v_\alpha m}[f].
\label{eq:transport_moment_theorem}
\end{equation}
In the Hermite-function realization this becomes
\begin{equation}
\sqrt n\,\langle M_m|
\widehat{\mathcal K}_{\rm p}|\Psi\rangle
=
\partial_t\!\left(\sqrt n\,\langle M_m|\Psi\rangle\right)
+\partial_\alpha\!\left(
\sqrt n\,\langle M_{v_\alpha m}|\Psi\rangle\right).
\label{eq:transport_moment_theorem_H}
\end{equation}
\end{theorem}

\begin{proof}
Using Eq.~\eqref{eq:physical_transport_intertwining},
\begin{align}
\mathcal M_m\!\left[
\mathcal Q\widehat{\mathcal K}_{\rm p}|\Psi\rangle
\right]
&=\int m(\bm v)\,\mathcal T f\,d^dv
\nonumber\\
&=\partial_t\int m f\,d^dv
+\partial_\alpha\int v_\alpha m f\,d^dv,
\end{align}
which is Eq.~\eqref{eq:transport_moment_theorem}.
\end{proof}

The theorem explains a cancellation that is otherwise easy to discover only
through explicit projection.  The differential connection may be complicated
and realization dependent, but its explicit form is immaterial for moment
transport once the complete propagation operator is intertwined with the
physical reconstruction map.  No compatibility condition has entered the
argument; the local parameters are still arbitrary smooth fields specifying
the parameterized representation.

If the observable itself depends explicitly on those fields, the same
calculation gives
\begin{equation}
\mathcal M_m\!\left[
\mathcal Q\widehat{\mathcal K}_{\rm p}|\Psi\rangle
\right]
=
\partial_t\mathcal M_m[f]
+\partial_\alpha\mathcal M_{v_\alpha m}[f]
-\mathcal M_{\mathcal Tm}[f].
\label{eq:transport_moment_moving_observable}
\end{equation}

\subsection{Hydrodynamic Moments and Fluxes}
\label{sec:hydrodynamic_moments}

We now specialize the general moment formula
\eqref{eq:moment_vacuum_pairing} to the moments and fluxes entering the
hydrodynamic balance equations. At this stage, the fields $n(\bm x,t)$, $\bm u(\bm x,t)$ and $\theta(\bm x,t)$ entering the Fock-space representation of the LFH operator 
are still regarded as parameters of the local realization. No
identification with the physical moments of the represented state is
assumed.
For convenience, introduce
\begin{equation}\label{eq:hydro_velocity_operator}
\begin{aligned}
&\widehat X_\alpha
=
\hat a_\alpha+\hat a_\alpha^\dagger,\\
&\widehat v_\alpha
=
u_\alpha I+\sqrt{\theta}\widehat X_\alpha,\\
&\widehat{\bm X}^2=\sum_{\alpha=1}^d\widehat X_\alpha\widehat X_\alpha
\end{aligned}
\end{equation}
The left-vacuum identities needed below are
\begin{align}
\langle0|\widehat X_\alpha
&=
\langle\bm e_\alpha|,
\label{eq:left_vacuum_X}
\\
\langle0|\widehat X_\alpha\widehat X_\beta
&=
\delta_{\alpha\beta}\langle0|
+
\langle0|\hat a_\alpha\hat a_\beta .
\label{eq:left_vacuum_XX}\\
\langle0|
\widehat{\bm X}^2
\widehat X_\alpha
&=
(d+2)\langle\bm e_\alpha|
+
\sqrt6\langle3\bm e_\alpha|
+
\sqrt2
\sum_{\beta\neq\alpha}
\langle2\bm e_\beta+\bm e_\alpha|,
\label{eq:left_vacuum_xxx}
\end{align}
We also expand the Fock state as
\begin{equation}
|\Psi\rangle
=
\sum_{\bm n}
\psi_{\bm n}|\bm n\rangle .
\label{eq:hydro_fock_expansion}
\end{equation}
Finally, it is useful to introduce the second-level Fock tensor
\begin{equation}
\Xi_{\alpha\beta}
:=\langle0|\hat a_\alpha\hat a_\beta|\Psi\rangle
=
\begin{cases}
\psi_{\bm e_\alpha+\bm e_\beta},
&\alpha\neq\beta,\\
\sqrt2\psi_{2\bm e_\alpha},
&\alpha=\beta.
\end{cases}
\label{eq:second_level_matrix_elements}
\end{equation}


\subsubsection{Particle density and particle flux}

Choosing $m(\bm v)=1$ in
Eq.~\eqref{eq:moment_vacuum_pairing} gives the particle-number density
\begin{equation}
N
:=
\int_{\mathbb R^d} f\,d^dv
=
\sqrt n\langle0|\Psi\rangle .
\label{eq:number_density_fock}
\end{equation}
Hence
\begin{equation}
N
=
\sqrt n\psi_{\bm0}.
\label{eq:number_density_amplitude}
\end{equation}
The particle-number flux is
\begin{equation}
J_\alpha
:=
\int_{\mathbb R^d}
v_\alpha f\,d^dv
=
\sqrt n
\langle0|\widehat v_\alpha|\Psi\rangle .
\label{eq:number_flux_fock}
\end{equation}
Using Eq.~\eqref{eq:hydro_velocity_operator},
\begin{align}
J_\alpha
&=
\sqrt n
\left[
u_\alpha\langle0|
+
\sqrt\theta\langle\bm e_\alpha|
\right]
|\Psi\rangle
\nonumber\\
&=
u_\alpha N
+
\sqrt{n\theta}
\psi_{\bm e_\alpha}.
\label{eq:number_flux_amplitude}
\end{align}
Thus,
\begin{equation}
J_\alpha-u_\alpha N = \sqrt{n\theta}\psi_{\bm e_\alpha}.
\label{eq:number_flux_mismatch}
\end{equation}
The momentum density is obtained by including the particle mass:
\begin{equation}
P_\alpha
:=
mJ_\alpha
=
m\int_{\mathbb R^d}
v_\alpha f\,d^dv .
\label{eq:momentum_density_definition}
\end{equation}
Consequently,
\begin{equation}
P_\alpha
-
mNu_\alpha
=
m\sqrt{n\theta}
\psi_{\bm e_\alpha}.
\label{eq:momentum_density_amplitude}
\end{equation}
Equation~\eqref{eq:momentum_density_amplitude} shows that the first Fock
level $|{\bm e_\alpha}\rangle$ measures the difference between the physical momentum density and
the momentum associated with the velocity parameter of the local
realization.


\subsubsection{Momentum-flux and pressure tensors}

The physical momentum-flux tensor is
\begin{equation}
\Pi_{\alpha\beta}
:=
m
\int_{\mathbb R^d}
v_\alpha v_\beta f\,d^dv
=
m\sqrt n
\langle0|
\widehat v_\alpha\widehat v_\beta
|\Psi\rangle .
\label{eq:momentum_flux_definition}
\end{equation}
Expanding the velocity operators and using the definition \eqref{eq:second_level_matrix_elements}, we get
\begin{equation}
\begin{aligned}
\Pi_{\alpha\beta}
={}&
mNu_\alpha u_\beta
+
m\sqrt{n\theta}
\left(
u_\alpha\psi_{\bm e_\beta}
+
u_\beta\psi_{\bm e_\alpha}
\right)
+
mN\theta\delta_{\alpha\beta}
+
m\sqrt n\theta
\Xi_{\alpha\beta}.
\end{aligned}
\label{eq:momentum_flux_compact}
\end{equation}

The second peculiar-velocity moment relative to the realization velocity is defined by
\begin{equation}
p_{\alpha\beta}
:=
m
\int_{\mathbb R^d}
c_\alpha c_\beta fd^dv.
\label{eq:pressure_tensor_definition}
\end{equation}
Since
\[
\widehat c_\alpha
=
\sqrt\theta
\widehat X_\alpha,
\]
its Fock-space representation is
\begin{align}
p_{\alpha\beta}
&=
m\sqrt n\theta
\langle0|
\widehat X_\alpha\widehat X_\beta
|\Psi\rangle
\nonumber\\
&=
mN\theta\delta_{\alpha\beta}
+
m\sqrt n\theta
\Xi_{\alpha\beta}.
\label{eq:pressure_tensor_fock}
\end{align}
Thus,
\begin{equation}
p_{\alpha\beta}
=
mN\theta\delta_{\alpha\beta}
+
m\sqrt n\theta
\Xi_{\alpha\beta}.
\label{eq:pressure_tensor_boxed}
\end{equation}
The trace of the second-level tensor $\Xi_{\alpha\beta}$ measures the difference between
the peculiar kinetic energy of the represented state and the temperature
parameter of the realization, while its traceless part represents the
non-equilibrium deviatoric stress.
It is useful to decompose this second peculiar-velocity moment into its rotationally
invariant scalar and traceless parts. Define the scalar pressure by
\begin{equation}
p
:=
\frac{1}{d}
p_{\gamma\gamma},
\label{eq:scalar_pressure_definition}
\end{equation}
and the trace-free deviatoric stress tensor by
\begin{equation}
\pi_{\alpha\beta}
:=
p_{\alpha\beta}
-
p\delta_{\alpha\beta},
\qquad
\pi_{\alpha\alpha}=0.
\label{eq:deviatoric_stress_definition}
\end{equation}
Using the Fock-space expression \eqref{eq:pressure_tensor_boxed}, the scalar pressure becomes
\begin{equation}
{
p
=
mN\theta
+
\frac{m\sqrt n\theta}{d}
\Xi_{\gamma\gamma}.
}
\label{eq:scalar_pressure_fock}
\end{equation}
The traceless stress is therefore
\begin{equation}
\pi_{\alpha\beta}
=
m\sqrt n\theta
\left(
\Xi_{\alpha\beta}
-
\frac{1}{d}
\Xi_{\gamma\gamma}\delta_{\alpha\beta}
\right).
\label{eq:deviatoric_stress_fock}
\end{equation}
Finally, the trace of the second-level Fock tensor \eqref{eq:second_level_matrix_elements} is
\begin{equation}
\Xi_{\gamma\gamma}
=
\sqrt2
\sum_{\gamma=1}^{d}
\psi_{2\bm e_\gamma}.
\label{eq:second_level_trace_pressure}
\end{equation}
Consequently,
\begin{equation}
p-mN\theta
=
\frac{m\sqrt{2n}\theta}{d}
\sum_{\gamma=1}^{d}
\psi_{2\bm e_\gamma}.
\label{eq:scalar_pressure_difference}
\end{equation}
Thus, the scalar trace of the second Fock level measures the mismatch
between this scalar peculiar-velocity moment and the value determined by the
temperature parameter of the local realization.  Once momentum compatibility
is imposed, $\bm u$ is the physical mean velocity and $p_{\alpha\beta}$
becomes the physical pressure tensor.  Its traceless part is then the
non-equilibrium deviatoric stress and is not eliminated by the temperature
compatibility condition introduced below.


\subsubsection{Total and internal kinetic energy}

The total kinetic-energy density is
\begin{equation}
E
:=
\frac{m}{2}
\int_{\mathbb R^d}
v_\alpha v_\alpha f\,d^dv
=
\frac{m}{2}\sqrt n
\langle0|\widehat v_\alpha\widehat v_\alpha|\Psi\rangle .
\label{eq:total_energy_definition}
\end{equation}
Using the preceding identities, we obtain
\begin{equation}
\begin{aligned}
E-
\frac{m}{2}
\left(
u_\alpha u_\alpha+d\theta
\right)N
=
m\sqrt{n\theta}
u_\alpha\psi_{\bm e_\alpha}
+
\frac{m}{2}
\sqrt{2n}\theta
\sum_{\alpha=1}^{d}
\psi_{2\bm e_\alpha}.
\end{aligned}
\label{eq:total_energy_explicit}
\end{equation}
The peculiar or internal kinetic-energy density relative to the velocity
parameter $\bm u$ is
\begin{equation}
U
:=
\frac{m}{2}
\int_{\mathbb R^d}
c_\alpha c_\alpha f\,d^dv
=
\frac12 p_{\alpha\alpha}.
\label{eq:internal_energy_definition}
\end{equation}
Therefore, using \eqref{eq:scalar_pressure_difference}, we have
\begin{equation}
U-\frac{md}{2}N\theta
=
\frac{m}{2}
\sqrt{2n}\theta
\sum_{\alpha=1}^{d}
\psi_{2\bm e_\alpha}.
\label{eq:internal_energy_mismatch}
\end{equation}
Thus, the scalar trace of the second Fock level measures the mismatch
between the peculiar kinetic energy relative to the realization velocity
and the value determined by the temperature parameter.  It becomes the
physical internal-energy mismatch once momentum compatibility is imposed.


\subsubsection{Total energy flux and heat flux}

The total kinetic-energy flux is
\begin{equation}
Q_\alpha
:=
\frac{m}{2}
\int_{\mathbb R^d}
v_\beta v_\beta v_\alpha f\,d^dv
=
\frac{m}{2}\sqrt n
\langle0|
\widehat v_\beta\widehat v_\beta
\widehat v_\alpha
|\Psi\rangle .
\label{eq:total_energy_flux_definition}
\end{equation}
Rather than expanding the operator expression directly in ladder
operators, it is useful first to separate the convective and peculiar
contributions. Introduce the first peculiar moment
\begin{equation}
J_\alpha^{(c)}
:=
\int_{\mathbb R^d}
c_\alpha f\,d^dv,
\qquad
c_\alpha:=v_\alpha-u_\alpha.
\label{eq:first_peculiar_moment}
\end{equation}
Using the definitions of $J_\alpha$ and $N$, this moment satisfies
\begin{equation}
J_\alpha^{(c)}
=J_\alpha-Nu_\alpha
=
\sqrt{n\theta}
\psi_{\bm e_\alpha}.
\label{eq:first_peculiar_moment_fock}
\end{equation}

The peculiar heat flux is defined by
\begin{equation}
q_\alpha
:=
\frac{m}{2}
\int_{\mathbb R^d}
c_\beta c_\beta c_\alpha f\,d^dv.
\label{eq:heat_flux_definition}
\end{equation}
Expanding $v_\alpha=u_\alpha+c_\alpha$ in the total energy flux gives
\begin{align}
Q_\alpha
=&
\frac{m}{2}Nu^2u_\alpha
+
\frac{m}{2}u^2J_\alpha^{(c)}
+
mu_\alpha u_\beta J_\beta^{(c)}
+
u_\beta p_{\alpha\beta}
+
Uu_\alpha
+
q_\alpha.
\label{eq:energy_flux_expansion}
\end{align}
The total energy density may be written as
\begin{equation}
E
=
\frac{m}{2}Nu^2
+
mu_\beta J_\beta^{(c)}
+
U.
\label{eq:energy_density_peculiar_decomposition}
\end{equation}
Consequently, the exact energy-flux decomposition is
\begin{equation}
Q_\alpha
=
Eu_\alpha
+
u_\beta p_{\alpha\beta}
+
q_\alpha
+
\frac{m}{2}u^2J_\alpha^{(c)}.
\label{eq:energy_flux_general_decomposition}
\end{equation}
Equation~\eqref{eq:energy_flux_general_decomposition} is valid without
identifying the velocity parameter $\bm u$ with the physical mean
velocity of the represented state. The additional term is determined
entirely by the first peculiar moment, or equivalently by the
first-level Fock amplitude through
Eq.~\eqref{eq:first_peculiar_moment_fock}. The condition under which this
term vanishes will be formulated in the following subsection.

Finally, the heat flux possesses the Fock representation
\begin{equation}
q_\alpha
=
\frac{m}{2}
\sqrt n\theta^{3/2}
\langle0|
\widehat X_\beta\widehat X_\beta
\widehat X_\alpha
|\Psi\rangle .
\label{eq:heat_flux_fock_operator}
\end{equation}
Using Eq.~\eqref{eq:second_level_matrix_elements}, we obtain
\begin{equation}
\begin{aligned}
q_\alpha
=
\frac{m}{2}
\sqrt n\theta^{3/2}
\bigg[
&(d+2)\psi_{\bm e_\alpha}
+
\sqrt6\psi_{3\bm e_\alpha}
+
\sqrt2
\sum_{\beta\neq\alpha}
\psi_{2\bm e_\beta+\bm e_\alpha}
\bigg].
\end{aligned}
\label{eq:heat_flux_explicit}
\end{equation}
The heat flux therefore contains a first-level contribution whenever the
local realization is not centered on the physical mean velocity. Once
momentum compatibility is imposed, the first-level contribution
vanishes, and the heat flux is carried entirely by the third Fock level.


\subsubsection{Fock hierarchy of hydrodynamic moments}

The preceding formulas establish a direct correspondence between the
hydrodynamic moment hierarchy and the first few levels of Fock space:
\begin{equation}
\begin{array}{ccl}
|\bm k|=0
&\longleftrightarrow &
\text{particle density},\\
|\bm k|=1
&\longleftrightarrow&
\text{peculiar flux, momentum mismatch, and off-compatible heat-flux term},
\\
|\bm k|=2
&\longleftrightarrow&
\text{peculiar-energy trace and second-moment tensor},
\\
|\bm k|=3
&\longleftrightarrow&
\text{third-order heat-flux contribution}.
\end{array}
\label{eq:hydrodynamic_fock_hierarchy}
\end{equation}
At this stage these relations are purely kinematic consequences of the
moment intertwining formula and the chosen Hermite-function realization.
No conservation law or compatibility condition has yet been imposed.
The relation between the parameters $n,\bm u,\theta$ defining the local
realization and the physical moments $N,\bm P,E$ of the represented
state is addressed in the following subsection.

\subsubsection{Summary of Hydrodynamic Moment Bras}

The general Hermite-function realization of a polynomial moment functional is
\[
\langle M_m|
:=
\langle0|\,m(\widehat{\bm v}),
\]
with
\[
\widehat v_\alpha
=
u_\alpha I
+
\sqrt{\theta}
\left(
\hat a_\alpha+\hat a_\alpha^\dagger
\right).
\]

The hydrodynamic moment and flux bras used below are
\begin{align}
\langle N|
&:=\langle0|,
\\
\langle P_\alpha|
&:=m\langle0|\widehat v_\alpha,
\\
\langle E|
&:=\frac{m}{2}\langle0|\widehat{\bm v}^{\,2},
\\
\langle U|
&:=\frac{m}{2}\langle0|(\widehat{\bm v}-\bm u)^2
=\frac{m}{2}\langle0|\widehat{\bm c}^{\,2},
\\
\langle \Pi_{\alpha\beta}|
&:=m\langle0|\widehat v_\alpha\widehat v_\beta,
\\
\langle Q_\alpha|
&:=\frac{m}{2}\langle0|\widehat{\bm v}^{\,2}\widehat v_\alpha,
\\
\langle q_\alpha|
&:=\frac{m}{2}\langle0|\widehat{\bm c}^{\,2}\widehat c_\alpha,
\end{align}
where
\[
\widehat c_\alpha
:=
\widehat v_\alpha-u_\alpha
=
\sqrt{\theta}
\left(
\hat a_\alpha+\hat a_\alpha^\dagger
\right).
\]
These bras are the Hermite-function representatives of the corresponding
physical observables.  In the next subsection they are applied to the
parameterized Fock relation $\widehat{\mathcal K}[\lambda]|\Psi\rangle=0$
before any compatibility condition is imposed.  The resulting exact moment
relations are then restricted to the compatibility set and, conversely, used
to propagate compatibility dynamically.

\subsection{Compatibility and Its Exact Propagation}
\label{sec:compatibility_exact}

The local realization is parameterized by $\lambda=(n,\bm u,\theta)$,
whereas the represented state has the independent physical moments
$N$, $\bm P$, and $E$.  The parameterized Fock relation
\begin{equation}
\widehat{\mathcal K}[\lambda]|\Psi\rangle=0
\label{eq:parameterized_fock_relation}
\end{equation}
is meaningful before these quantities are identified.  Compatibility selects
the locally adapted representation.

\subsubsection{Algebraic compatibility}

We impose
\begin{equation}
N=n,
\qquad
P_\alpha=mnu_\alpha,
\qquad
E=E_\lambda
:=\frac m2nu_\alpha u_\alpha+\frac{md}{2}n\theta .
\label{eq:fock_compatibility_physical}
\end{equation}
Equivalently, using the Hermite-function moment formulas,
\begin{equation}
\psi_{\bm0}=\sqrt n,
\qquad
\psi_{\bm e_\alpha}=0,
\qquad
\sum_{\alpha=1}^{d}\psi_{2\bm e_\alpha}=0.
\label{eq:fock_compatibility_coefficients}
\end{equation}
The first relation fixes the vacuum amplitude, the second removes the first
Fock level, and the third removes only the scalar trace component of level
$2$.  The traceless level-$2$ sector and all higher levels remain available to
carry nonequilibrium stress, heat flux, and higher moments.

On the compatibility manifold the exact flux decompositions simplify to
\begin{equation}
Q_\alpha
=Eu_\alpha+u_\beta p_{\alpha\beta}+q_\alpha,
\label{eq:energy_flux_filtered}
\end{equation}
and the heat flux contains only level-$3$ amplitudes,
\begin{equation}
q_\alpha
=\frac m2\sqrt n\,\theta^{3/2}
\left[
\sqrt6\,\psi_{3\bm e_\alpha}
+\sqrt2\sum_{\beta\ne\alpha}
\psi_{2\bm e_\beta+\bm e_\alpha}
\right].
\label{eq:heat_flux_filtered}
\end{equation}

\subsubsection{Moment equations off compatibility}

The transport contributions follow immediately from
Theorem~\ref{thm:transport_moment_intertwining}.  The remaining
Ornstein--Uhlenbeck moments are local and may be evaluated algebraically.  For
the parameterized Fock relation one obtains the exact identities
\begin{align}
\partial_tN+\partial_\alpha J_\alpha
&=0,
\label{eq:fock_offcompat_density}
\\
\partial_tP_\beta+\partial_\alpha\Pi_{\beta\alpha}
+\gamma\left(P_\beta-mu_\beta N\right)
&=0,
\label{eq:fock_offcompat_momentum}
\\
\partial_tE+\partial_\alpha Q_\alpha
+\gamma\left(2E-u_\alpha P_\alpha-md\theta N\right)
&=0.
\label{eq:fock_offcompat_energy}
\end{align}
These equations are valid before compatibility.  They also provide a useful
check on the sign convention: with
$\widehat{\mathcal K}_{\rm OU}=+\gamma\widehat{\mathcal N}$, the production
terms relax the momentum and energy mismatches toward the local parameters.
On the compatibility manifold the production terms vanish and
Eqs.~\eqref{eq:fock_offcompat_density}--\eqref{eq:fock_offcompat_energy}
reduce to the exact conservative LFH balances.

This gives one implementation of self-consistency: impose
Eq.~\eqref{eq:fock_compatibility_physical} algebraically at every time and the
balance laws are consequences.  The converse, dynamical implementation is
more useful for the present representation.

\subsubsection{Dynamic compatibility propagation}

Define the compatibility residuals
\begin{equation}
C_0:=N-n,
\qquad
C_\beta:=P_\beta-mnu_\beta,
\qquad
C_E:=E-E_\lambda.
\label{eq:compatibility_residuals}
\end{equation}
Suppose the realization parameters are evolved with the conservative balances
obtained on the compatibility manifold,
\begin{align}
\partial_tn+\partial_\alpha(nu_\alpha)&=0,
\label{eq:parameter_density_balance}
\\
\partial_t(mnu_\beta)+\partial_\alpha\Pi_{\beta\alpha}&=0,
\label{eq:parameter_momentum_balance}
\\
\partial_tE_\lambda+\partial_\alpha Q_\alpha&=0.
\label{eq:parameter_energy_balance}
\end{align}
Subtracting these equations from the off-compatible moment equations gives a
closed homogeneous system for the residuals,
\begin{align}
\partial_tC_0+\frac1m\partial_\alpha C_\alpha&=0,
\label{eq:compat_residual_density}
\\
\partial_tC_\beta
+\gamma\left(C_\beta-mu_\beta C_0\right)&=0,
\label{eq:compat_residual_momentum}
\\
\partial_tC_E
+\gamma\left(2C_E-u_\alpha C_\alpha-md\theta C_0\right)&=0.
\label{eq:compat_residual_energy}
\end{align}

\begin{theorem}[Exact compatibility propagation]
\label{thm:compatibility_propagation}
Let $|\Psi\rangle$ satisfy the parameterized Fock relation
Eq.~\eqref{eq:parameterized_fock_relation}, and let
$\lambda=(n,\bm u,\theta)$ evolve according to
Eqs.~\eqref{eq:parameter_density_balance}--\eqref{eq:parameter_energy_balance}.
If $C_0=C_\alpha=C_E=0$ initially, then the compatibility residuals remain
zero for as long as the coupled solution exists.
\end{theorem}

\begin{proof}
Equations~\eqref{eq:compat_residual_density}--\eqref{eq:compat_residual_energy}
form a homogeneous linear system in the residuals along the prescribed
solution.  The zero residual is therefore propagated identically.
\end{proof}

No closure approximation has been used.  The theorem shows that the local
realization need not be re-fitted algebraically at every time: its parameters
may instead be evolved dynamically so that the compatibility manifold is
invariant.  For the LFH model the algebraic and dynamical routes are
equivalent; the latter formulation is especially suggestive for more general
parameterized representations in which compatibility may be implicit rather
than explicitly solvable.

The higher Fock amplitudes are not constrained by compatibility.  They enter
the exact stress and heat fluxes in the parameter balances.  Their
asymptotic evaluation is the subject of the hydrodynamic construction in the
next section.

\section{Compatibility Propagation in the Hydrodynamic Limit}\label{sec:Chapman_Enskog}

The previous section established the exact propagation of the compatibility conditions. We now investigate their perturbative realization in the hydrodynamic limit. The resulting hierarchy is organized according to the grading of the Fock space induced by the number operator.

\subsection{Vacuum Grading and Selection Rules}
\label{sec:vacuum_grading_selection}

We now turn from the exact compatibility formulation to its perturbative
realization in the hydrodynamic limit.  The calculation will be carried
out entirely in Fock space.  The Hermite-function realization enters
only through the already established differential connectors,
\begin{equation}
\widehat{\mathcal C}_{u_\beta}
=
\frac{1}{2\sqrt{\theta}}
\left(
\hat a_\beta^\dagger-\hat a_\beta
\right),
\qquad
\widehat{\mathcal C}_{\theta}
=
\frac{1}{4\theta}
\left[
\left(\hat a_\beta^\dagger\right)^2-\hat a_\beta^2
\right].
\label{eq:connectors_sec6_again}
\end{equation}
No explicit Hermite function will be required below.

The occupation-number basis carries the natural multigrading by
$\bm n=(n_1,\ldots,n_d)$.  For the present purpose we use the coarser
grading by the total occupation number,
\begin{equation}
|\bm n|
=
\sum_{\alpha=1}^{d}n_\alpha,
\qquad
\mathscr H
=
\bigoplus_{k=0}^{\infty}\mathscr H_k,
\qquad
\mathscr H_k
=
{\rm span}\{
|\bm n\rangle:|\bm n|=k
\}.
\label{eq:total_fock_grading}
\end{equation}
This is the grading naturally adapted to the Ornstein--Uhlenbeck
relaxation, since
\begin{equation}
\widehat{\mathcal N}
=
\hat a_\alpha^\dagger\hat a_\alpha,
\qquad
\widehat{\mathcal N}|_{\mathscr H_k}=kI.
\label{eq:number_grading_sec6}
\end{equation}

An operator $A$ is called homogeneous of grading $r$ if
\begin{equation}
A:\mathscr H_k\longrightarrow\mathscr H_{k+r}
\label{eq:grading_definition_sec6}
\end{equation}
for every $k$ for which the action is defined.  Equivalently,
\begin{equation}
[\widehat{\mathcal N},A]=rA.
\label{eq:grading_commutator_sec6}
\end{equation}
Hence
\begin{equation}
{\rm gr}(\hat a_\alpha)=-1,
\qquad
{\rm gr}(\hat a_\alpha^\dagger)=+1.
\label{eq:ladder_grading_sec6}
\end{equation}

The grading is additive under products.  Indeed, if
$[\widehat{\mathcal N},A]=rA$ and
$[\widehat{\mathcal N},B]=sB$, then
\begin{align}
[\widehat{\mathcal N},AB]
&=
[\widehat{\mathcal N},A]B
+
A[\widehat{\mathcal N},B]
\nonumber\\
&=
(r+s)AB,
\end{align}
and therefore
\begin{equation}
{\rm gr}(AB)={\rm gr}(A)+{\rm gr}(B).
\label{eq:grading_product_rule}
\end{equation}
A sum of monomials need not itself be homogeneous; it is understood as
the sum of its homogeneous components.

For the first hydrodynamic correction it is unnecessary to construct
all homogeneous components of the full propagation operator on an
arbitrary Fock state.  The zeroth-order state will be the vacuum state,
and the relevant object is therefore the \emph{vacuum action} of the
elementary operator blocks.  With
\[
\widehat X_\alpha
=
\hat a_\alpha+\hat a_\alpha^\dagger,
\qquad
\hat a_\alpha|0\rangle=0,
\]
the canonical commutation relations give
\begin{align}
\widehat X_\alpha|0\rangle
&=
\hat a_\alpha^\dagger|0\rangle,
\label{eq:vacuum_rule_X}
\\
\widehat{\mathcal C}_{u_\beta}|0\rangle
&=
\frac{1}{2\sqrt{\theta}}\,
\hat a_\beta^\dagger|0\rangle,
\label{eq:vacuum_rule_Cu}
\\
\widehat{\mathcal C}_{\theta}|0\rangle
&=
\frac{1}{4\theta}
\left(\hat a_\beta^\dagger\right)^2|0\rangle,
\label{eq:vacuum_rule_Ctheta}
\\
\widehat X_\alpha
\widehat{\mathcal C}_{u_\beta}|0\rangle
&=
\frac{1}{2\sqrt{\theta}}
\left(
\delta_{\alpha\beta}
+
\hat a_\alpha^\dagger\hat a_\beta^\dagger
\right)|0\rangle,
\label{eq:vacuum_rule_XCu}
\\
\widehat X_\alpha
\widehat{\mathcal C}_{\theta}|0\rangle
&=
\frac{1}{4\theta}
\left[
2\hat a_\alpha^\dagger
+
\hat a_\alpha^\dagger
\left(\hat a_\beta^\dagger\right)^2
\right]|0\rangle,
\label{eq:vacuum_rule_XCtheta}
\\
\widehat X_\alpha\widehat X_\beta|0\rangle
&=
\left(
\delta_{\alpha\beta}
+
\hat a_\alpha^\dagger\hat a_\beta^\dagger
\right)|0\rangle,
\label{eq:vacuum_rule_XX}
\\
\widehat{\bm X}^{\,2}|0\rangle
&=
\left[
d+
\left(\hat a_\beta^\dagger\right)^2
\right]|0\rangle,
\label{eq:vacuum_rule_X2}
\\
\widehat{\bm X}^{\,2}\widehat X_\alpha|0\rangle
&=
\left[
(d+2)\hat a_\alpha^\dagger
+
\hat a_\alpha^\dagger
\left(\hat a_\beta^\dagger\right)^2
\right]|0\rangle.
\label{eq:vacuum_rule_X2X}
\end{align}

Equations~\eqref{eq:vacuum_rule_X}--\eqref{eq:vacuum_rule_X2X}
constitute the vacuum-selection table used below.  They provide an
economical alternative to expanding the complete noncommuting
propagation operator before the moment projection is specified.  Once compatibility has removed the hydrodynamic components, the
nonequilibrium moment bras of interest select specific irreducible tensor
sectors, and only operator blocks whose vacuum action reaches those sectors
need be evaluated.  For later reference, the two nonequilibrium observables
of interest select
\begin{equation}
\begin{aligned}
\text{deviatoric stress:}\qquad
&\mathscr H_{2,\mathrm{tl}},
\\
\text{heat flux:}\qquad
&\mathscr H_{3,\mathrm v}\subset\mathscr H_3,
\end{aligned}
\label{eq:constitutive_level_selection}
\end{equation}
where $\mathscr H_{3,\mathrm v}$ denotes the vector (trace) component of
the third Fock level, spanned by the states
$\hat a_\alpha^\dagger(\hat a_\beta^\dagger)^2|0\rangle$ with the
repeated $\beta$ summed.  The level-$1$ part of the heat-flux moment is
absent on the compatibility manifold.

\subsection{Hydrodynamic Scaling and the First-Order Compatibility Source}
\label{sec:hydrodynamic_scaling_fock}

With the complete propagation part,
\begin{equation}
\widehat{\mathcal K}_{\rm p}
:=
\widehat{\mathcal K}_{\rm tr}
+
\widehat{\mathcal K}_{\rm con}.
\label{eq:Kp_definition}
\end{equation}
together with the Ornstein--Uhlenbeck relaxation operator
\begin{equation}
\widehat{\mathcal K}_{\rm OU}
:=
\gamma\widehat{\mathcal N},
\label{eq:positive_OU_relaxation}
\end{equation}
the hydrodynamic scaling is then written as
\begin{equation}
\varepsilon\,
\widehat{\mathcal K}_{\rm p}|\Psi\rangle
+
\gamma\widehat{\mathcal N}|\Psi\rangle
=
0,
\label{eq:scaled_fock_equation}
\end{equation}
where $\varepsilon$ is the small Knudsen-type parameter.

We seek
\begin{equation}
|\Psi\rangle
=
|\Psi^{(0)}\rangle
+
\varepsilon|\Psi^{(1)}\rangle
+
O(\varepsilon^2).
\label{eq:fock_CE_expansion}
\end{equation}
The macroscopic fields are the compatible fields of the representation;
equivalently, the higher-order corrections are required not to modify
their defining moments.  This is equivalent, at the present order, to
the usual Chapman--Enskog convention in which the material derivatives
appearing in the first kinetic source are evaluated with the
Euler-order evolution obtained from the solvability/compatibility
conditions below.  At zeroth order,
\begin{equation}
|\Psi^{(0)}\rangle
=
\sqrt n\,|0\rangle.
\label{eq:zeroth_order_vacuum}
\end{equation}
The leading equation is identically satisfied because
$\widehat{\mathcal N}|0\rangle=0$.  At first order,
\begin{equation}
\gamma\widehat{\mathcal N}|\Psi^{(1)}\rangle
=
-
|\Sigma^{(0)}\rangle,
\qquad
|\Sigma^{(0)}\rangle
:=
\widehat{\mathcal K}_{\rm p}
\sqrt n\,|0\rangle.
\label{eq:first_order_fock_equation}
\end{equation}
The source $|\Sigma^{(0)}\rangle$ is the Fock-space analogue of the
classical Chapman--Enskog propagation of the local Maxwellian.  Its
direct evaluation as a full operator expression is unnecessarily long.
The vacuum-selection rules above allow it to be assembled level by
level.

We shall use
\begin{equation}
D_u
=
\partial_t+u_\alpha\partial_\alpha.
\label{eq:Du_sec6}
\end{equation}

Using Eqs.~\eqref{eq:vacuum_rule_X}--\eqref{eq:vacuum_rule_X2X}, the
transport sector gives
\begin{equation}
\begin{aligned}
\widehat{\mathcal K}_{\rm tr}
\sqrt n\,|0\rangle
=
\frac{\sqrt n}{2}
\Bigg\{
&
A_0|0\rangle
+
A_\alpha
\hat a_\alpha^\dagger|0\rangle
+
B_{\alpha\beta}
\hat a_\alpha^\dagger
\hat a_\beta^\dagger|0\rangle
+
\Theta_\alpha
\hat a_\alpha^\dagger
\left(\hat a_\beta^\dagger\right)^2
|0\rangle
\Bigg\},
\end{aligned}
\label{eq:Ktr_vacuum_compact}
\end{equation}
where
\begin{align}
A_0
&:=
\frac{D_u n}{n}
+
\partial_\alpha u_\alpha,
\label{eq:A0_definition}
\\
A_\alpha
&:=
\frac{1}{\sqrt{\theta}}
\left[
D_u u_\alpha
+
\frac{1}{n}\partial_\alpha(n\theta)
\right],
\label{eq:A1_definition}
\\
B_{\alpha\beta}
&:=
\frac12
\left(
\partial_\alpha u_\beta
+
\partial_\beta u_\alpha
\right)
+
\frac{D_u\theta}{2\theta}
\delta_{\alpha\beta},
\label{eq:B2_definition}
\\
\Theta_\alpha
&:=
\frac{1}{2\sqrt{\theta}}
\partial_\alpha\theta.
\label{eq:Theta3_definition}
\end{align}

The same selection rules may be applied directly to
the local-Maxwellian connector.  Before collecting Fock levels, its
vacuum action is
\begin{align}
\widehat{\mathcal K}_{\rm con}\sqrt n\,|0\rangle
=
\sqrt n\Bigg\{&
\left[
\frac{D_u n}{2n}
-
\frac{d}{4\theta}D_u\theta
\right]
+
\widehat X_\alpha
\left[
\frac{\sqrt{\theta}}{2n}\partial_\alpha n
-
\frac{d}{4\sqrt{\theta}}\partial_\alpha\theta
+
\frac{D_u u_\alpha}{2\sqrt{\theta}}
\right]
\nonumber\\
&+
\frac12\widehat X_\alpha\widehat X_\beta
\partial_\alpha u_\beta
+
\frac{D_u\theta}{4\theta}
\widehat{\bm X}^{\,2}
+
\frac{\partial_\alpha\theta}{4\sqrt{\theta}}
\widehat{\bm X}^{\,2}\widehat X_\alpha
\Bigg\}|0\rangle .
\label{eq:Kcon_vacuum_before_selection}
\end{align}
The scalar pieces proportional to $dD_u\theta/(4\theta)$ cancel;
the level-$1$ temperature terms combine according to
\[
-\frac{d}{4\sqrt{\theta}}
+
\frac{d+2}{4\sqrt{\theta}}
=
\frac{1}{2\sqrt{\theta}},
\]
while the level-$2$ and level-$3$ components are read off directly from
Eqs.~\eqref{eq:vacuum_rule_XX}--\eqref{eq:vacuum_rule_X2X}.  The result is
\begin{equation}
\widehat{\mathcal K}_{\rm con}
\sqrt n\,|0\rangle
=
\widehat{\mathcal K}_{\rm tr}
\sqrt n\,|0\rangle.
\label{eq:equal_half_sources}
\end{equation}

The equality is not accidental.  At local equilibrium the physical
Maxwellian is split into two identical square-root factors, one carried
by the similarity transformation and one by the coordinate vacuum.
Differentiating the full Maxwellian differentiates both factors, and the
two equal half-contributions recombine.  In particular, the two cubic
thermal-gradient terms are born to add rather than cancel.

Consequently, the complete first-order propagation source is
\begin{equation}
\begin{aligned}
|\Sigma^{(0)}\rangle
=
\sqrt n
\Bigg[
&
A_0
+
A_\alpha\hat a_\alpha^\dagger
+
B_{\alpha\beta}
\hat a_\alpha^\dagger\hat a_\beta^\dagger
+
\Theta_\alpha
\hat a_\alpha^\dagger
\left(\hat a_\beta^\dagger\right)^2
\Bigg]|0\rangle.
\end{aligned}
\label{eq:first_order_source_full}
\end{equation}
Thus the complete source generated from the compatible vacuum occupies
only Fock levels $0$ through $3$.  No explicit Hermite polynomial or
Gaussian integral has entered the calculation.

\subsection{Hydrodynamic Null Projections and Euler Compatibility}
\label{sec:hydro_null_projection}

Equation~\eqref{eq:first_order_fock_equation} can preserve the
compatibility constraints only if the source has no component along the
Fock sectors defining the hydrodynamic fields.  This gives an immediate
and stringent check on all connector coefficients.

The density projection is obtained with the already defined moment bra $\langle N|=\langle0|$:
Using Eq.~\eqref{eq:first_order_source_full},
\begin{equation}
\sqrt n\,\langle N|\Sigma^{(0)}\rangle
=
nA_0
=
D_u n+n\partial_\alpha u_\alpha.
\label{eq:density_null_check}
\end{equation}
Hence the density projection vanishes precisely when the Euler
continuity equation is satisfied,
\begin{equation}
D_u n+n\partial_\alpha u_\alpha=0.
\label{eq:Euler_continuity_sec6}
\end{equation}

The momentum projection uses $\langle P_\mu|=m\langle0|\widehat v_\mu$:
Its projection is
\begin{equation}
\sqrt n\,\langle P_\mu|\Sigma^{(0)}\rangle
=
mn
\left(
u_\mu A_0
+
\sqrt{\theta}\,A_\mu
\right).
\label{eq:momentum_projection_before_continuity}
\end{equation}
Using the Euler continuity relation from the density projection,
\begin{equation}
\sqrt n\,\langle P_\mu|\Sigma^{(0)}\rangle
=
mn
\left[
D_u u_\mu
+
\frac{1}{n}\partial_\mu(n\theta)
\right].
\label{eq:momentum_null_check}
\end{equation}
Thus the first-level source vanishes precisely for
\begin{equation}
D_u u_\mu
+
\frac{1}{n}\partial_\mu(n\theta)
=
0,
\label{eq:Euler_momentum_sec6}
\end{equation}
which is the unforced Euler momentum equation.

For the energy projection we use $\langle E|=(m/2)\langle0|\widehat{\bm v}^{\,2}$:
A direct evaluation gives
\begin{equation}
\begin{aligned}
\sqrt n\,\langle E|\Sigma^{(0)}\rangle
=
\frac{mn}{2}
\Big[
&
(u^2+d\theta)A_0
+
2\sqrt{\theta}\,u_\alpha A_\alpha
+
2\theta B_{\alpha\alpha}
\Big].
\end{aligned}
\label{eq:energy_projection_general_sec6}
\end{equation}
After the density and momentum projections have vanished, only the
scalar level-$2$ component remains:
\begin{equation}
\sqrt n\,\langle E|\Sigma^{(0)}\rangle
=
mn\theta
\left[
\partial_\alpha u_\alpha
+
\frac{d}{2\theta}D_u\theta
\right].
\label{eq:energy_null_check}
\end{equation}
Hence the energy projection vanishes precisely when
\begin{equation}
D_u\theta
+
\frac{2}{d}\theta
\partial_\alpha u_\alpha
=
0.
\label{eq:Euler_temperature_sec6}
\end{equation}

Equations~\eqref{eq:density_null_check},
\eqref{eq:momentum_null_check}, and
\eqref{eq:energy_null_check} provide the desired hydrodynamic null
check:
\begin{equation}
\sqrt n\,\langle N|\Sigma^{(0)}\rangle
=
\sqrt n\,\langle P_\alpha|\Sigma^{(0)}\rangle
=
\sqrt n\,\langle E|\Sigma^{(0)}\rangle
=0
\label{eq:hydrodynamic_null_check_box}
\end{equation}
when the Euler-order compatibility equations hold.

This result should not be interpreted as an independent closure
condition imposed from outside the kinetic equation.  It is the
zeroth-order realization of the exact compatibility propagation
established in the preceding section.  In particular, the Euler
equations remove respectively the level-$0$, level-$1$, and scalar
level-$2$ components of the first-order Fock source.

Using Eq.~\eqref{eq:Euler_temperature_sec6}, the tensor
$B_{\alpha\beta}$ becomes traceless:
\begin{equation}
B_{\alpha\beta}
=
\mathsf S_{\alpha\beta},
\qquad
\mathsf S_{\alpha\beta}
:=
\frac12
\left(
\partial_\alpha u_\beta
+
\partial_\beta u_\alpha
\right)
-
\frac{1}{d}
\delta_{\alpha\beta}
\partial_\gamma u_\gamma.
\label{eq:strain_tensor_sec6}
\end{equation}
Therefore the source restricted to the kinetic complement reduces to
\begin{equation}
|\Sigma^{(0)}_{\rm kin}\rangle
=
\sqrt n
\left[
\mathsf S_{\alpha\beta}
\hat a_\alpha^\dagger\hat a_\beta^\dagger
+
\frac{\partial_\alpha\theta}{2\sqrt{\theta}}
\hat a_\alpha^\dagger
\left(\hat a_\beta^\dagger\right)^2
\right]
|0\rangle.
\label{eq:kinetic_source_after_Euler}
\end{equation}

\subsection{First Kinetic Correction and Navier--Stokes--Fourier Pairings}
\label{sec:first_kinetic_correction}

The decisive simplification now comes from the spectrum of the number
operator.  On the two sectors present in
Eq.~\eqref{eq:kinetic_source_after_Euler},
\begin{equation}
\widehat{\mathcal N}|_{\mathscr H_2}=2I,
\qquad
\widehat{\mathcal N}|_{\mathscr H_3}=3I.
\label{eq:N_inverse_levels_23}
\end{equation}
Equation~\eqref{eq:first_order_fock_equation} can therefore be inverted
level by level without constructing any infinite operator matrix.  One
obtains
\begin{equation}
\begin{aligned}
|\Psi^{(1)}\rangle
=
-\sqrt n
\Bigg[
&
\frac{1}{2\gamma}
\mathsf S_{\alpha\beta}
\hat a_\alpha^\dagger\hat a_\beta^\dagger
+
\frac{1}{6\gamma\sqrt{\theta}}
(\partial_\alpha\theta)
\hat a_\alpha^\dagger
\left(\hat a_\beta^\dagger\right)^2
\Bigg]|0\rangle .
\end{aligned}
\label{eq:Psi1_complete_NSF}
\end{equation}

Equation~\eqref{eq:Psi1_complete_NSF} makes the compatibility
propagation explicit.  The first correction contains neither a vacuum
amplitude nor a level-$1$ amplitude, and its level-$2$ component is
traceless.  Hence
\begin{equation}
\langle0|\Psi^{(1)}\rangle=0,
\qquad
\langle\bm e_\alpha|\Psi^{(1)}\rangle=0,
\qquad
\sum_{\alpha}
\langle2\bm e_\alpha|\Psi^{(1)}\rangle=0.
\label{eq:Psi1_compatibility_check}
\end{equation}

\subsubsection{Deviatoric stress}

The irreducible stress pairing established previously is
\begin{equation}
\pi_{\mu\nu}^{(1)}
=
m\sqrt n\,\theta
\left\langle0\left|
\hat a_\mu\hat a_\nu
-
\frac{1}{d}
\delta_{\mu\nu}
\hat a_\gamma\hat a_\gamma
\right|\Psi^{(1)}\right\rangle.
\label{eq:stress_pairing_first_order}
\end{equation}
Only the level-$2$ term in Eq.~\eqref{eq:Psi1_complete_NSF} can
contribute.  The required vacuum contraction is
\begin{equation}
\langle0|
\hat a_\mu\hat a_\nu
\hat a_\alpha^\dagger\hat a_\beta^\dagger
|0\rangle
=
\delta_{\mu\alpha}\delta_{\nu\beta}
+
\delta_{\mu\beta}\delta_{\nu\alpha}.
\label{eq:stress_vacuum_contraction}
\end{equation}
Since $\mathsf S_{\alpha\beta}$ is symmetric and traceless,
Eq.~\eqref{eq:stress_pairing_first_order} becomes
\begin{equation}
\pi_{\mu\nu}^{(1)}
=
-
\frac{mn\theta}{\gamma}
\mathsf S_{\mu\nu}.
\label{eq:NS_stress_final}
\end{equation}
In the conventional notation
\begin{equation}
\pi_{\mu\nu}^{(1)}
=
-2\mu\,\mathsf S_{\mu\nu},
\end{equation}
the dynamic viscosity is therefore
\begin{equation}
\mu
=
\frac{mn\theta}{2\gamma}.
\label{eq:viscosity_OU_Fock}
\end{equation}

The selection mechanism is worth emphasizing.  Before evaluating the
matrix element, the stress bra eliminates every vacuum-action block
except those carrying a traceless level-$2$ component.  In the original
operator these are precisely the tensorial pieces generated by
$\widehat X_\alpha\widehat{\mathcal C}_{u_\beta}$ and
$\widehat X_\alpha\widehat X_\beta$.  Scalar level-$2$ contributions
from $\widehat{\mathcal C}_\theta$ and
$\widehat{\bm X}^{\,2}$ are removed by the traceless projection.

\subsubsection{Heat flux}

Under momentum compatibility, the heat-flux pairing is supported on
level $3$:
\begin{equation}
q_\mu^{(1)}
=
\frac{m}{2}
\sqrt n\,\theta^{3/2}
\left\langle0\left|
\widehat{\bm X}^{\,2}\widehat X_\mu
\right|\Psi^{(1)}\right\rangle .
\label{eq:heat_pairing_first_order}
\end{equation}
Only the cubic term of Eq.~\eqref{eq:Psi1_complete_NSF} contributes.
The required contraction is
\begin{equation}
\left\langle0\left|
\hat a_\gamma\hat a_\gamma\hat a_\mu
\,
\hat a_\alpha^\dagger
\left(\hat a_\beta^\dagger\right)^2
\right|0\right\rangle
=
2(d+2)\delta_{\mu\alpha},
\label{eq:heat_vacuum_contraction}
\end{equation}
where summation over $\beta$ and $\gamma$ is understood.  It follows
that
\begin{equation}
q_\mu^{(1)}
=
-
\frac{d+2}{6\gamma}
mn\theta\,
\partial_\mu\theta.
\label{eq:Fourier_flux_final}
\end{equation}
Thus
\begin{equation}
q_\mu^{(1)}
=
-\kappa\,\partial_\mu\theta,
\qquad
\kappa
=
\frac{d+2}{6\gamma}
mn\theta.
\label{eq:thermal_conductivity_OU_Fock}
\end{equation}

Here again the vacuum-selection rule avoids the full operator matrix
element.  The only two propagation blocks capable of reaching level $3$
from the vacuum are
$\widehat X_\alpha\widehat{\mathcal C}_{\theta}$ from the differential
transport connector and
$\widehat{\bm X}^{\,2}\widehat X_\alpha$ from the local-Maxwellian
connector.   Their cubic vacuum contributions
are equal and add.  This doubling is the operator counterpart of
differentiating the two square-root Gaussian factors whose product is
the local Maxwellian.

As an independent consistency check, the resulting transport coefficients yield the Prandtl number ${\rm Pr}=3/2$, independent of the dimension $d$, in agreement with the conventional Chapman--Enskog analysis of the original LFH kinetic equation. In the Fock-space realization, this simply says that the heat flux lives on the Fock level $3$ and the deviatoric stress on the Fock level $2$.

\subsubsection{Navier--Stokes--Fourier balance equations}

Substitution of Eqs.~\eqref{eq:NS_stress_final} and
\eqref{eq:Fourier_flux_final} into the exact compatible balance
equations gives the first hydrodynamic correction.  The momentum
equation becomes
\begin{equation}
\partial_t(mnu_\alpha)
+
\partial_\beta
\left[
mnu_\alpha u_\beta
+
mn\theta\,\delta_{\alpha\beta}
-
\frac{mn\theta}{\gamma}
\mathsf S_{\alpha\beta}
\right]
=
0.
\label{eq:NS_momentum_Fock}
\end{equation}
The internal-energy form is
\begin{equation}
\frac{md}{2}nD_u\theta
+
mn\theta\,\partial_\alpha u_\alpha
-
\frac{mn\theta}{\gamma}
\mathsf S_{\alpha\beta}
\partial_\alpha u_\beta
-
\partial_\alpha
\left[
\frac{d+2}{6\gamma}
mn\theta\,\partial_\alpha\theta
\right]
=
0.
\label{eq:NSF_energy_Fock}
\end{equation}
These equations are understood to the order retained in the
hydrodynamic expansion.

\subsection{Interpretation: Compatibility Propagation in the Hydrodynamic Limit}
\label{sec:compatibility_hydro_interpretation}

The calculation above is the perturbative realization of the exact
compatibility propagation established in the preceding section.  The
logical sequence is
\begin{equation}
\begin{array}{c}
\text{compatible vacuum}
\\
\downarrow
\\
\text{propagation source}
\\
\downarrow
\\
\text{hydrodynamic projections}
\;\Rightarrow\;
\text{Euler compatibility}
\\
\downarrow
\\
\text{first-order kinetic source}
\subset
\mathscr H_{2,\mathrm{tl}}\oplus\mathscr H_{3,\mathrm v}
\\
\downarrow
\\
\text{level-wise OU inversion}
\\
\downarrow
\\
\text{Navier--Stokes--Fourier pairings}.
\end{array}
\label{eq:compatibility_hydro_flow}
\end{equation}

In the classical Chapman--Enskog calculation, the first-order
Maxwellian propagation is simplified through an extensive set of
velocity-polynomial identities and cancellations.  The corresponding
Fock-space propagation operator is algebraically longer because the
velocity products are replaced by noncommuting ladder operators.
However, the Fock representation supplies a different economy.  The
vacuum, total Fock grading, and irreducible moment bras provide an
a priori selection rule: only operator blocks whose vacuum action
reaches the sector detected by the prescribed observable need be
opened.

Thus the computational sequence is not
\[
\text{expand the complete operator}
\longrightarrow
\text{simplify the resulting expression},
\]
but rather
\begin{equation}
\begin{aligned}
&\text{choose perturbative order}
\longrightarrow\text{choose moment pairing}
\longrightarrow\text{select Fock sector},\\
&\hspace{34mm}\longrightarrow\text{evaluate only the contributing operator blocks}.
\end{aligned}
\label{eq:Fock_selection_strategy}
\end{equation}

The role of the Hermite realization is correspondingly different from
that in a Grad-type expansion.  The Hermite functions were used once,
at the construction stage, to establish the coordinate realization and
derive the differential connectors.  Their complete effect is then
encoded algebraically in
$\widehat{\mathcal C}_{u_\alpha}$ and
$\widehat{\mathcal C}_{\theta}$.  No explicit Hermite function,
Hermite recurrence relation, or Gaussian quadrature enters
Eqs.~\eqref{eq:first_order_source_full}--\eqref{eq:NSF_energy_Fock}.
The hydrodynamic-limit calculation is therefore performed entirely
within Fock space.

This also clarifies the relation to the classical Chapman--Enskog
procedure.  The latter is not introduced here as an external closure
prescription for otherwise unclosed balance equations.  Rather, it is
the perturbative evaluation of the same compatibility propagation that
holds exactly in the full Fock-space kinetic equation.  At zeroth order
compatibility yields the Euler equations; at first order the kinetic
complement generates the viscous stress and heat flux.  In this sense,
the hydrodynamic expansion is a perturbative realization of
compatibility propagation, no more and no less.

\section{Covariance under Changes of Coordinate Realization}
\label{sec:covariance}

The preceding sections have established three ingredients of the
Fock-space formulation. First, the transformed kinetic generator admits
an exact pull-back to the abstract Fock space. Second, the macroscopic
moments are represented by linear functionals on that space, and
compatibility identifies the parameters of the local realization with
the corresponding moments of the represented state. Third, the exact
propagation of compatibility and its hydrodynamic-limit realization can
be formulated entirely in terms of the abstract Fock dynamics.

These results were derived using a particular local coordinate
realization. The purpose of the present section is to show that the
resulting Fock-space formulation is not tied to that choice. We consider
two admissible coordinate realizations related by a local invertible
similarity transformation and establish covariance at four levels:
the kinetic generator, the pull-back of external derivatives, the moment
functionals, and the propagation of compatibility. The abstract
Fock-space objects remain unchanged, while their coordinate
representatives transform covariantly.


\subsection{Equivalent Coordinate Realizations}
\label{sec:covariance_representation_equivalence}

Let $\mathscr H$ denote the abstract Fock space and let
$\mathscr F_1$ and $\mathscr F_2$ be two admissible coordinate
realization spaces. For fixed local macroscopic parameters
\[
\lambda(\bm x,t)
:=
\bigl(n(\bm x,t),\bm u(\bm x,t),\theta(\bm x,t)\bigr),
\]
let
\begin{equation}
\mathcal R_i(\bm x,t):
\mathscr H\longrightarrow\mathscr F_i,
\qquad i=1,2,
\label{eq:cov_Ri}
\end{equation}
be invertible realization maps on the domains under consideration.
The coordinate representatives of an abstract state
$|\Psi(\bm x,t)\rangle\in\mathscr H$ are
\begin{equation}
\Phi_i
=
\mathcal R_i|\Psi\rangle.
\label{eq:cov_state_realizations}
\end{equation}

Since both realization maps are invertible, they determine a unique
local isomorphism
\begin{equation}
\mathcal S
:=
\mathcal R_2\mathcal R_1^{-1}
:
\mathscr F_1\longrightarrow\mathscr F_2,
\label{eq:cov_S_definition}
\end{equation}
so that
\begin{equation}
\mathcal R_2
=
\mathcal S\mathcal R_1,
\qquad
\Phi_2
=
\mathcal S\Phi_1.
\label{eq:cov_R2_SR1}
\end{equation}
In general $\mathcal S=\mathcal S(\bm x,t)$ because the coordinate
realizations depend locally on the macroscopic fields.

\begin{proposition}[Representation equivalence]
\label{prop:representation_equivalence}
Let $\mathcal O_1$ and $\mathcal O_2$ be coordinate operators acting
on $\mathscr F_1$ and $\mathscr F_2$, respectively. They represent the
same abstract Fock-space operator $\widehat{\mathcal O}$ if and only if
\begin{equation}
\mathcal O_2
=
\mathcal S\mathcal O_1\mathcal S^{-1},
\label{eq:cov_operator_similarity}
\end{equation}
where the composition in
Eq.~\eqref{eq:cov_operator_similarity} includes the action of
differential operators on the local map $\mathcal S$.
In that case
\begin{equation}
\widehat{\mathcal O}
=
\mathcal R_1^{-1}\mathcal O_1\mathcal R_1
=
\mathcal R_2^{-1}\mathcal O_2\mathcal R_2.
\label{eq:cov_operator_invariance}
\end{equation}
\end{proposition}

\begin{proof}
Using $\mathcal R_2=\mathcal S\mathcal R_1$ and
Eq.~\eqref{eq:cov_operator_similarity},
\begin{align}
\mathcal R_2^{-1}\mathcal O_2\mathcal R_2
&=
\mathcal R_1^{-1}\mathcal S^{-1}
\left(
\mathcal S\mathcal O_1\mathcal S^{-1}
\right)
\mathcal S\mathcal R_1
\nonumber\\
&=
\mathcal R_1^{-1}\mathcal O_1\mathcal R_1.
\end{align}
Conversely, equality of the two pull-backs implies
Eq.~\eqref{eq:cov_operator_similarity}.
\end{proof}

The statement is elementary for operators acting only on the local
velocity coordinate. Its nontrivial content for the present kinetic
problem concerns operators containing the external derivatives
$\partial_t$ and $\partial_\alpha$. Because $\mathcal S$ is local,
those derivatives do not commute with the similarity map. Their
transformation is considered next.

\begin{remark}[Covariance versus invariance]
The coordinate representatives $\mathcal O_1$ and $\mathcal O_2$ are
\emph{covariant}: they change according to
Eq.~\eqref{eq:cov_operator_similarity}. The pulled-back Fock operator
$\widehat{\mathcal O}$ is \emph{invariant}: it is the same abstract
operator in both realizations. We use covariance for the complete
representation-theoretic statement containing both properties.
\end{remark}

\subsection{Covariance of External Derivatives and Differential Connectors}
\label{sec:covariance_external_derivatives}

The treatment of external derivatives is the essential extension of
the purely algebraic similarity argument. Let $\partial_\mu$ denote
either $\partial_t$ or one of the spatial derivatives
$\partial_\alpha$. In realization $i$ define the differential connector
\begin{equation}
\widehat{\mathcal G}^{(i)}_\mu
:=
\mathcal R_i^{-1}
\left(
\partial_\mu\mathcal R_i
\right).
\label{eq:cov_Gi_definition}
\end{equation}
As established in the general pull-back construction,
\begin{equation}
\mathcal R_i^{-1}
\circ\partial_\mu\circ
\mathcal R_i
=
\partial_\mu
+
\widehat{\mathcal G}^{(i)}_\mu .
\label{eq:cov_derivative_pullback_i}
\end{equation}


The local similarity map introduces the coordinate-space connection
\begin{equation}
\Omega_\mu
:=
\left(
\partial_\mu\mathcal S
\right)\mathcal S^{-1},
\label{eq:cov_Omega_definition}
\end{equation}
which acts on $\mathscr F_2$. Conjugation of the bare derivative gives
\begin{align}
\mathcal S\partial_\mu\mathcal S^{-1}
&=
\partial_\mu
+
\mathcal S
\left(
\partial_\mu\mathcal S^{-1}
\right)
\nonumber\\
&=
\partial_\mu
-
\left(
\partial_\mu\mathcal S
\right)\mathcal S^{-1}.
\end{align}
Hence
\begin{equation}
\mathcal D_\mu^{(2)}
:=
\mathcal S\partial_\mu\mathcal S^{-1}
=
\partial_\mu-\Omega_\mu.
\label{eq:cov_D2}
\end{equation}
Thus the bare derivative is not itself invariant under a local
similarity transformation. The transformed coordinate derivative
contains an induced connection term.

\begin{proposition}[Transformation law of the differential connector]
\label{prop:cov_connector_transformation}
If $\mathcal R_2=\mathcal S\mathcal R_1$, then
\begin{equation}
\widehat{\mathcal G}^{(2)}_\mu
=
\widehat{\mathcal G}^{(1)}_\mu
+
\widehat{\Lambda}_\mu,
\label{eq:cov_G_transformation}
\end{equation}
where
\begin{equation}
\widehat{\Lambda}_\mu
:=
\mathcal R_1^{-1}
\mathcal S^{-1}
\left(
\partial_\mu\mathcal S
\right)
\mathcal R_1
=
\mathcal R_2^{-1}
\Omega_\mu
\mathcal R_2.
\label{eq:cov_Lambda_definition}
\end{equation}
\end{proposition}

\begin{proof}
Differentiating $\mathcal R_2=\mathcal S\mathcal R_1$ gives
\[
\partial_\mu\mathcal R_2
=
(\partial_\mu\mathcal S)\mathcal R_1
+
\mathcal S(\partial_\mu\mathcal R_1).
\]
Multiplication from the left by
$\mathcal R_2^{-1}=\mathcal R_1^{-1}\mathcal S^{-1}$ yields
Eq.~\eqref{eq:cov_G_transformation}.
\end{proof}

Equations~\eqref{eq:cov_D2} and
\eqref{eq:cov_G_transformation} combine to give
\begin{align}
\mathcal R_2^{-1}
\mathcal D_\mu^{(2)}
\mathcal R_2
&=
\mathcal R_2^{-1}
\left(
\partial_\mu-\Omega_\mu
\right)
\mathcal R_2
\nonumber\\
&=
\partial_\mu
+
\widehat{\mathcal G}^{(2)}_\mu
-
\widehat{\Lambda}_\mu
\nonumber\\
&=
\partial_\mu
+
\widehat{\mathcal G}^{(1)}_\mu
\nonumber\\
&=
\mathcal R_1^{-1}
\partial_\mu
\mathcal R_1.
\end{align}
Therefore
\begin{equation}
\mathcal R_2^{-1}
\left(
\mathcal S\partial_\mu\mathcal S^{-1}
\right)
\mathcal R_2
=
\mathcal R_1^{-1}
\partial_\mu
\mathcal R_1.
\label{eq:cov_derivative_invariance}
\end{equation}

The inhomogeneous change of the realization connector is therefore
exactly compensated by the connection generated through conjugation of
the coordinate derivative.

\begin{remark}[Bare derivative versus covariant coordinate derivative]
Pulling back the \emph{bare} derivative through two different
realization maps gives
\[
\mathcal R_i^{-1}\partial_\mu\mathcal R_i
=
\partial_\mu+\widehat{\mathcal G}^{(i)}_\mu,
\]
and these expressions generally differ. By contrast, the coordinate
operator representing the same abstract derivative in realization $2$
is
$\mathcal D_\mu^{(2)}=
\mathcal S\partial_\mu\mathcal S^{-1}$.
Its pull-back is identical to that of the derivative in realization
$1$.
\end{remark}

The same mechanism applies to the quasi-linear differential terms of
the kinetic generator. Let
\begin{equation}
\mathcal O^{(1)}_\mu
=
P_1\partial_\mu,
\qquad
P_2
=
\mathcal S P_1\mathcal S^{-1}.
\end{equation}
Then
\begin{equation}
\mathcal O^{(2)}_\mu
=
\mathcal S
\mathcal O^{(1)}_\mu
\mathcal S^{-1}
=
P_2
\left(
\partial_\mu-\Omega_\mu
\right).
\label{eq:cov_quasilinear_transformed}
\end{equation}
The ordering is fixed: $P_2$ acts to the left of the transformed
derivative. If
\[
\widehat P
=
\mathcal R_1^{-1}P_1\mathcal R_1
=
\mathcal R_2^{-1}P_2\mathcal R_2,
\]
then
\begin{align}
\mathcal R_2^{-1}
\mathcal O^{(2)}_\mu
\mathcal R_2
&=
\widehat P
\left(
\partial_\mu
+
\widehat{\mathcal G}^{(2)}_\mu
-
\widehat{\Lambda}_\mu
\right)
\nonumber\\
&=
\widehat P
\left(
\partial_\mu
+
\widehat{\mathcal G}^{(1)}_\mu
\right)
\nonumber\\
&=
\mathcal R_1^{-1}
\mathcal O^{(1)}_\mu
\mathcal R_1.
\end{align}
Thus the covariance theorem extends to the transport sector and to the
quasi-linear differential operators used in the LFH generator.

\begin{remark}[Connection-like transformation law]
Equation~\eqref{eq:cov_G_transformation} is an inhomogeneous
transformation law characteristic of a connection. The differential
connector should therefore not be regarded as an additional physical
interaction: its form depends on the chosen local coordinate
realization, while the complete pulled-back derivative is invariant.
\end{remark}

\subsection{Covariance of the Complete Kinetic Generator}
\label{sec:covariance_complete_generator}

Let $\mathcal K_1$ be one admissible coordinate representation of the
kinetic dynamics and define
\begin{equation}
\mathcal K_2
=
\mathcal S\mathcal K_1\mathcal S^{-1},
\label{eq:cov_K2_definition}
\end{equation}
with all external derivatives understood as composition operators.
Local velocity-space sectors transform algebraically, whereas every
external derivative transforms according to
Eq.~\eqref{eq:cov_D2}. The local and propagation sectors therefore
transform by different intermediate rules but recombine into the same
abstract operator.

\begin{theorem}[Covariance of the Fock-space kinetic representation]
\label{thm:cov_fock_generator}
Let $\mathcal R_1$ and $\mathcal R_2$ be two admissible local
realization maps related by
$\mathcal R_2=\mathcal S\mathcal R_1$. Let their coordinate kinetic
generators be related by
$\mathcal K_2=\mathcal S\mathcal K_1\mathcal S^{-1}$, with the induced
transformation of all external derivatives included. Then
\begin{equation}
\widehat{\mathcal K}
=
\mathcal R_1^{-1}
\mathcal K_1
\mathcal R_1
=
\mathcal R_2^{-1}
\mathcal K_2
\mathcal R_2.
\label{eq:cov_complete_K_invariance}
\end{equation}
Moreover, if
\[
\mathcal K_1\Phi_1=0,
\qquad
\Phi_2=\mathcal S\Phi_1,
\]
then
\begin{equation}
\mathcal K_2\Phi_2
=
\mathcal S
\mathcal K_1\Phi_1
=
0.
\label{eq:cov_coordinate_equations}
\end{equation}
Thus the two coordinate equations represent the same abstract
Fock-space dynamics.
\end{theorem}

\begin{proof}
The operator identity follows from
Proposition~\ref{prop:representation_equivalence}, with
Sec.~\ref{sec:covariance_external_derivatives} providing the required
treatment of the external derivatives. The state equation follows by
direct substitution.
\end{proof}

\begin{remark}
The theorem is stronger than the statement that two local
velocity-space algebras are isomorphic. It includes the space--time
propagation sector, whose covariance depends essentially on the
inhomogeneous transformation of the differential connector.
\end{remark}

\subsection{Covariance of Physical Reconstruction and Moments}
\label{sec:covariance_moments}

The coordinate state by itself is not the physical kinetic state.  Let
$\mathcal B_i$ denote the reconstruction map associated with
$\mathcal R_i$, so that
\begin{equation}
f=\mathcal B_i\Phi_i,
\qquad
\Phi_i=\mathcal R_i|\Psi\rangle,
\qquad i=1,2.
\label{eq:cov_reconstruction_i}
\end{equation}
Since $\Phi_2=\mathcal S\Phi_1$, representation of the same physical
distribution requires
\begin{equation}
\mathcal B_2=\mathcal B_1\mathcal S^{-1}.
\label{eq:cov_B_transformation}
\end{equation}
Consequently the composite physical realization is invariant,
\begin{equation}
\mathcal Q
=\mathcal B_1\mathcal R_1
=\mathcal B_2\mathcal R_2,
\qquad
f=\mathcal Q|\Psi\rangle.
\label{eq:cov_Q_invariance}
\end{equation}
This is the physical counterpart of the operator identity
Eq.~\eqref{eq:cov_operator_invariance}: the coordinate factors change, while
the map from the abstract ket to the physical distribution does not.

For any velocity observable $m(\bm v)$,
\begin{equation}
\mathcal M_m[f]
=\int m(\bm v)f(\bm v)\,d^dv
=\mathcal M_m\!\left[\mathcal Q|\Psi\rangle\right]
\label{eq:cov_moment_invariance}
\end{equation}
is therefore independent of the chosen coordinate representative.  In a
particular realization the same scalar may be evaluated with a
realization-specific weight or bra, but those are coordinate devices rather
than additional invariant objects.  In the Hermite-function realization,
for example,
\begin{equation}
\mathcal M_m[f]
=\sqrt n\,\langle M_m|\Psi\rangle,
\qquad
\langle M_m|=\langle0|m(\widehat{\bm v}).
\label{eq:cov_connection_to_moment_bras}
\end{equation}

\subsection{Covariance of Compatibility and Its Propagation}
\label{sec:covariance_compatibility}

Compatibility is formulated entirely in terms of physical moments and the
common parameter fields.  If $A$ labels one of the hydrodynamic observables,
let $\mathcal H_A(\lambda)$ denote its prescribed value; for the LFH fields,
\begin{equation}
\mathcal H_N(\lambda)=n,
\qquad
\mathcal H_{P_\alpha}(\lambda)=mnu_\alpha,
\qquad
\mathcal H_E(\lambda)=E_\lambda
=\frac m2nu_\alpha u_\alpha+\frac{md}{2}n\theta.
\label{eq:cov_hydro_targets_LFH}
\end{equation}
Define
\begin{equation}
\mathfrak C_A[|\Psi\rangle;\lambda]
:=\mathcal M_A\!\left[\mathcal Q|\Psi\rangle\right]
-\mathcal H_A(\lambda).
\label{eq:cov_compatibility_residual}
\end{equation}
Because both $\mathcal Q$ and $\lambda$ are unchanged by a change of
coordinate realization, the compatibility residual is representation
invariant.  In particular,
\begin{equation}
\Phi_1\text{ is compatible with }\lambda
\quad\Longleftrightarrow\quad
\Phi_2=\mathcal S\Phi_1\text{ is compatible with }\lambda.
\label{eq:cov_compatibility_equivalence}
\end{equation}
The coefficient conditions
\begin{equation}
\psi_{\bm0}=\sqrt n,
\qquad
\psi_{\bm e_\alpha}=0,
\qquad
\sum_{\alpha=1}^{d}\psi_{2\bm e_\alpha}=0
\label{eq:cov_fock_compatibility_conditions}
\end{equation}
therefore refer to the common abstract ket and are not peculiar to the
Hermite-function coordinate state.

The same argument extends to compatibility propagation.  Let
$\mathcal E_i:=\mathcal K_i\Phi_i$ be the coordinate kinetic residual.
From $\mathcal K_2=\mathcal S\mathcal K_1\mathcal S^{-1}$ and
$\Phi_2=\mathcal S\Phi_1$,
\begin{equation}
\mathcal E_2=\mathcal S\mathcal E_1.
\label{eq:cov_residual_transformation}
\end{equation}
Using Eq.~\eqref{eq:cov_B_transformation},
\begin{equation}
\mathcal B_2\mathcal E_2
=\mathcal B_1\mathcal E_1.
\label{eq:cov_physical_residual_invariance}
\end{equation}
Hence every physical moment of the kinetic residual is the same in the two
realizations.  The exact balance equations obtained by hydrodynamic
projection, and therefore the homogeneous evolution equations for the
compatibility residuals, are identical.  Exact compatibility propagation is
thus a property of the represented kinetic dynamics, not of a particular
coordinate realization.

\subsection{Relation to the Hydrodynamic-Limit Construction}
\label{sec:covariance_hydrodynamic_limit}

The hydrodynamic calculation of Sec.~\ref{sec:Chapman_Enskog} was carried out
in Hermite-function coordinates because this realization is spectrally
adapted to the Ornstein--Uhlenbeck sector: its vacuum is explicit and the
number operator diagonalizes relaxation.  Covariance shows that this is a
computational privilege, not a physical restriction.  Stress and heat flux
are physical moments, and the Euler and Navier--Stokes--Fourier equations
obtained from compatibility propagation are independent of the coordinate
representative.

This statement should be distinguished from the occupation-number grading
itself.  The decomposition
\[
\mathscr H=\bigoplus_{k\ge0}\mathscr H_k
\]
is especially useful for the LFH relaxation operator because
$\widehat{\mathcal N}$ is diagonal on these sectors.  A different collision
operator may call for a different spectral or block organization; the
covariance statements concern the abstract dynamics and its physical
reconstruction, not the choice of grading used to calculate them.

The covariance structure may be summarized as
\begin{equation}
\begin{gathered}
\mathcal R_2=\mathcal S\mathcal R_1,
\qquad
\Phi_2=\mathcal S\Phi_1,
\qquad
\mathcal B_2=\mathcal B_1\mathcal S^{-1},
\\
\mathcal K_2=\mathcal S\mathcal K_1\mathcal S^{-1},
\\[1mm]
\mathcal Q=\mathcal B_i\mathcal R_i,
\qquad
\widehat{\mathcal K}
=\mathcal R_i^{-1}\mathcal K_i\mathcal R_i,
\end{gathered}
\label{eq:covariance_summary}
\end{equation}
where the last line is independent of $i$.  Physical moments,
compatibility, and compatibility propagation then follow from the common
physical state $f=\mathcal Q|\Psi\rangle$.

\subsection{Explicit Polynomial--Function Covariance}
\label{sec:covariance_poly_function_example}

The Hermite-polynomial and Hermite-function realizations provide a concrete
check of the preceding statements without introducing new structure.  From
Sec.~\ref{sec:coordinate_realizations_equivalence},
\begin{equation}
\mathcal R_H=\mathcal S_W\mathcal R_P,
\qquad
\mathcal S_W\Phi=\sqrt W\,\Phi,
\qquad
\Phi_H=\sqrt W\,\Phi_P.
\label{eq:cov_basis_relation_sec7}
\end{equation}
Their physical reconstructions are correspondingly
\begin{equation}
f=\sqrt n\,\sqrt W\,\Phi_H
=\sqrt n\,W\,\Phi_P.
\label{eq:cov_reconstruction_HP}
\end{equation}
Thus $\mathcal B_H=\sqrt n\sqrt W$ and
$\mathcal B_P=\sqrt n W$, while
\begin{equation}
\mathcal B_H\mathcal R_H
=\mathcal B_P\mathcal R_P
=\mathcal Q.
\label{eq:cov_Q_HP}
\end{equation}
Equation~\eqref{eq:cov_reconstruction_HP} is also the direct explanation for
why the moment weights differ in the two coordinate systems although the
physical moments do not.

For the Weyl generators,
\begin{align}
\mathscr a_\alpha^{(P)}
&=\partial_{\xi_\alpha},
&
\mathscr a_\alpha^{(P)\dagger}
&=-\partial_{\xi_\alpha}+\xi_\alpha,
\label{eq:cov_P_ladders}
\\
\mathscr a_\alpha^{(H)}
&=\partial_{\xi_\alpha}+\frac{\xi_\alpha}{2},
&
\mathscr a_\alpha^{(H)\dagger}
&=-\partial_{\xi_\alpha}+\frac{\xi_\alpha}{2},
\label{eq:cov_H_ladders}
\end{align}
and Eq.~\eqref{eq:cov_basis_relation_sec7} gives the conjugation relations
$\mathscr a^{(H)}=\mathcal S_W\mathscr a^{(P)}\mathcal S_W^{-1}$ and
likewise for creation.  The coordinate number operators are therefore
\begin{align}
\mathcal N_P
&=-\partial_{\xi_\alpha}\partial_{\xi_\alpha}
+\xi_\alpha\partial_{\xi_\alpha},
\label{eq:cov_NP}
\\
\mathcal N_H
&=-\partial_{\xi_\alpha}\partial_{\xi_\alpha}
+\frac{\xi_\alpha\xi_\alpha}{4}-\frac d2,
\label{eq:cov_NH}
\end{align}
so that
\begin{equation}
\mathcal R_P^{-1}(\gamma\mathcal N_P)\mathcal R_P
=\mathcal R_H^{-1}(\gamma\mathcal N_H)\mathcal R_H
=\gamma\widehat{\mathcal N}.
\label{eq:cov_OU_explicit}
\end{equation}
The two differential representatives are different, but they encode the same
graded Fock relaxation operator.

The local similarity also acts nontrivially on external derivatives.  Since
$\mathcal S_W$ is multiplication by $\sqrt W$,
\begin{equation}
\Omega_\mu^{(W)}
:=(\partial_\mu\mathcal S_W)\mathcal S_W^{-1}
=\frac{\partial_\mu u_\beta}{2\sqrt\theta}\,\xi_\beta
+\frac{\partial_\mu\theta}{4\theta}
\left(\xi_\beta\xi_\beta-d\right),
\label{eq:cov_Omega_W}
\end{equation}
and hence
\begin{equation}
\mathcal S_W\partial_\mu\mathcal S_W^{-1}
=\partial_\mu-\Omega_\mu^{(W)}.
\label{eq:cov_derivative_W}
\end{equation}
At the Fock level the corresponding realization connectors satisfy the
general inhomogeneous transformation law of
Eq.~\eqref{eq:cov_G_transformation}; the extra term generated by
Eq.~\eqref{eq:cov_derivative_W} compensates it in the complete propagation
operator.  Consequently,
\begin{equation}
\mathcal K_H
=\mathcal S_W\mathcal K_P\mathcal S_W^{-1},
\qquad
\mathcal R_H^{-1}\mathcal K_H\mathcal R_H
=\mathcal R_P^{-1}\mathcal K_P\mathcal R_P
=\widehat{\mathcal K}.
\label{eq:cov_complete_generator_HP}
\end{equation}

The example therefore tests all parts of the construction with very different
coordinate representatives: the vacuum, Weyl generators, differential
connection, coordinate relaxation operator, and physical reconstruction all
change, whereas the abstract Fock generator and the reconstructed physical
distribution do not.  The square-root/Hermite-function realization is the
reference used in this paper because it makes the LFH relaxation sector
particularly simple.  Equivalent realizations inherit their reconstruction
from this reference through the known similarity map.  No claim is needed
here about coordinate representations that are not connected to this
similarity class.

\section{Conclusion}\label{sec:conclusion}

We have constructed an exact Fock-space representation of the force-free LFH
kinetic equation while keeping separate three logically distinct ingredients:
the parameterized kinetic generator, its coordinate realization, and the
compatibility between local parameters and physical moments.  The familiar
square-root Maxwellian transformation is used because it turns the LFH
relaxation into the number operator and supplies an elementary grading; the
representation-theoretic content begins with the subsequent intertwining of
the complete kinetic dynamics.

The main technical extension beyond the local oscillator algebra is the
pull-back of external derivatives.  A space--time dependent realization map
necessarily produces the differential connector
$\mathcal R^{-1}(\partial_\mu\mathcal R)$.  Its matrix elements are obtained
constructively by differentiating the local coordinate basis.  Once the full
propagation operator is intertwined, however, its moments no longer require
an explicit connector calculation: the transport--moment theorem follows
directly from the intertwining of the composite physical realization.  This
provides a substantial economy in the hydrodynamic projections and makes
clear that the cancellation between the differential and local-Maxwellian
connection sectors is structural rather than a peculiarity of Hermite
identities.

Compatibility was formulated both algebraically and dynamically.  Off the
compatibility manifold, the parameterized Fock relation yields exact moment
equations with relaxation production terms measuring the mismatch between the
represented state and the local parameters.  On compatibility these terms
vanish.  Conversely, if the local parameters are evolved with the compatible
balance laws, the compatibility residuals satisfy a homogeneous system and
remain zero exactly.  The dynamical formulation treats the motion of the local
realization as part of the kinetic dynamics and appears particularly suitable
for extensions in which an explicit moment fixed point may be unavailable.

In the hydrodynamic limit, the number operator provides more than a spectral
analogy with the harmonic oscillator: it grades the kinetic correction and
allows the relaxation inverse to be taken level by level.  Euler
compatibility removes the vacuum, level-$1$, and scalar level-$2$ components
of the first source.  The remaining traceless level-$2$ sector and the vector component of
level $3$ produce the viscous stress and heat flux, respectively, and recover the
Navier--Stokes--Fourier coefficients of the LFH model.  No explicit Hermite
functions are needed at this stage; their role has already been encoded in
the Fock operators and differential connection.

Finally, covariance under local changes of coordinate realization shows which
parts of the construction are intrinsic.  Hermite-polynomial and
Hermite-function coordinates have different vacua, differential generators,
connections, and moment weights, yet all pull back to the same abstract Fock
generator and the same physical moment functionals.  The square-root/Hermite-
function realization is therefore a convenient construction, not a privileged
physical description.

The LFH model was chosen precisely because its local relaxation sector is so
simple.  The structures that remain after this simplification---intertwining
of propagation, moving local realizations, dual moment functionals,
compatibility, and covariance---do not depend conceptually on the oscillator
interpretation of that sector.  This suggests studying the same architecture
for kinetic models with different local relaxation operators, where the
number grading may be replaced by another useful spectral or block
decomposition.

A natural next test is the nonlinear Boltzmann collision operator.  In
particular, Maxwell molecules offer a setting in which the propagation sector
is unchanged while the local collision sector is genuinely nonlinear and has
a nontrivial algebraic structure.  A Fock-space treatment of that case would
test how far the present intertwining, reconstruction, compatibility, and
grading framework extends beyond the diagonal LFH relaxation operator.  We
plan to address this problem in subsequent work.

\section*{Acknowledgement of AI assistance}
During the development and preparation of this manuscript, the author used
ChatGPT (OpenAI, GPT-5.6 Sol) as an interactive research and writing assistant.
Its use included discussion and critical examination of mathematical
formulations, cross-checking of derivations, development of the presentation
and organization of the manuscript, identification of notational and
expository inconsistencies, and assistance with editorial revision and
bibliographic verification.  All mathematical results, arguments,
interpretations, references, and final text were reviewed and accepted by the
author, who assumes full responsibility for the content of the manuscript.

\appendix




\section{Differential Connection in the Hermite-Function Realization}
\label{app:hermite_differential_connectors}

This appendix gives a detailed derivation and a set of independent consistency
checks for the differential connectors associated with the parameter-dependent
Hermite-function realization. Particular attention is paid to all numerical
prefactors, signs, normalization factors, and summation conventions. To simplify notation within the appendix, we write $\mathcal R=\mathcal R_H$ and drop the superscript $(H)$ on the Hermite-function coordinate basis and its differential connectors.

\subsection{Definitions and conventions}

Let
\begin{equation}
    c_\alpha = v_\alpha-u_\alpha,
    \qquad
    \xi_\alpha=\frac{c_\alpha}{\sqrt\theta},
    \label{eq:app_xi_definition}
\end{equation}
where the derivatives with respect to the macroscopic fields $u_\alpha$ and
$\theta$ are always taken at fixed physical velocity $\bm v$. Consequently,
\begin{equation}
    \frac{\partial \xi_\alpha}{\partial u_\beta}
    =-\frac{\delta_{\alpha\beta}}{\sqrt\theta},
    \qquad
    \frac{\partial \xi_\alpha}{\partial\theta}
    =-\frac{\xi_\alpha}{2\theta}.
    \label{eq:app_xi_parameter_derivatives}
\end{equation}

The local Maxwellian weight and its square root are
\begin{align}
    W
    &=
    \frac{1}{(2\pi\theta)^{d/2}}
    \exp\!\left(-\frac{c_\alpha c_\alpha}{2\theta}\right),
    \
    \sqrt W
    =
    \frac{1}{(2\pi\theta)^{d/4}}
    \exp\!\left(-\frac{\xi_\alpha\xi_\alpha}{4}\right).
    \label{eq:app_sqrtW}
\end{align}
Repeated Cartesian indices are summed from $1$ to $d$ unless stated otherwise.

The coordinate ladder operators are
\begin{equation}
    \mathcal A_\alpha
    =
    \frac{\partial}{\partial\xi_\alpha}
    +\frac{\xi_\alpha}{2},
    \qquad
    \mathcal A_\alpha^\dagger
    =
    -\frac{\partial}{\partial\xi_\alpha}
    +\frac{\xi_\alpha}{2},
    \label{eq:app_coordinate_ladders}
\end{equation}
so that
\begin{equation}
    [\mathcal A_\alpha,\mathcal A_\beta^\dagger]
    =\delta_{\alpha\beta},
    \qquad
    \mathcal A_\alpha\sqrt W=0.
\end{equation}
The inverse relations are
\begin{equation}
    \frac{\partial}{\partial\xi_\alpha}
    =\frac12\left(\mathcal A_\alpha-\mathcal A_\alpha^\dagger\right),
    \qquad
    \xi_\alpha
    =\mathcal A_\alpha+\mathcal A_\alpha^\dagger.
    \label{eq:app_inverse_ladder_relations}
\end{equation}

The normalized coordinate images of the Fock basis are
\begin{equation}
    \varphi_{\bm n}
    =
    \frac{1}{\sqrt{\bm n!}}
    \prod_{\alpha=1}^{d}
    \left(\mathcal A_\alpha^\dagger\right)^{n_\alpha}
    \sqrt W,
    \label{eq:app_basis_operator_form}
\end{equation}
where
\begin{equation}
    |\bm n|=\sum_{\alpha=1}^{d}n_\alpha,
    \qquad
    \bm n!=\prod_{\alpha=1}^{d}n_\alpha!.
\end{equation}
Equivalently,
\begin{equation}
    \varphi_{\bm n}(\bm\xi)
    =
    \sqrt W\,
    \frac{1}{\sqrt{2^{|\bm n|}\bm n!}}
    H_{\bm n}\!\left(\frac{\bm\xi}{\sqrt2}\right).
    \label{eq:app_basis_explicit_form}
\end{equation}
The normalization is chosen such that
\begin{equation}
   \langle \varphi_{\bm m},
    \varphi_{\bm n}\rangle=\int_{\mathbb R^d}
    \varphi_{\bm m}(\bm v)
    \varphi_{\bm n}(\bm v)
    \,d^dv
    =\delta_{\bm m,\bm n}.
    \label{eq:app_basis_orthonormality}
\end{equation}
Indeed, since $d^dv=\theta^{d/2}d^d\xi$,
\begin{equation}
    W\,d^dv
    =
    \frac{1}{(2\pi)^{d/2}}
    \exp\!\left(-\frac{\bm\xi^2}{2}\right)d^d\xi,
\end{equation}
and the change of variables $\bm y=\bm\xi/\sqrt2$ reduces
\eqref{eq:app_basis_orthonormality} to the standard orthogonality relation for
the physicists' Hermite polynomials.

The ladder action on the normalized basis is
\begin{align}
    \mathcal A_\alpha\varphi_{\bm n}
    &=
    \sqrt{n_\alpha}\,
    \varphi_{\bm n-\bm e_\alpha},
    \
    \mathcal A_\alpha^\dagger\varphi_{\bm n}
    =
    \sqrt{n_\alpha+1}\,
    \varphi_{\bm n+\bm e_\alpha},
    \label{eq:app_ladder_action}
\end{align}
where $\bm e_\alpha$ denotes the unit multi-index in direction $\alpha$.

Let $\mathcal R$ denote the realization map
\begin{equation}
    \mathcal R|\bm n\rangle=\varphi_{\bm n}.
\end{equation}
For any realization parameter $\lambda$, the corresponding differential
connector is
\begin{equation}
    \widehat{\mathcal C}_\lambda
    =
    \mathcal R^{-1}(\partial_\lambda\mathcal R).
    \label{eq:app_connector_definition}
\end{equation}
It is characterized by
\begin{equation}
    \partial_\lambda\varphi_{\bm n}
    =
    \mathcal R\widehat{\mathcal C}_\lambda|\bm n\rangle.
    \label{eq:app_connector_basis_characterization}
\end{equation}

\subsection{Velocity connector}

At fixed $\bm v$, dependence on $u_\beta$ enters only through $\bm\xi$.
Equation \eqref{eq:app_xi_parameter_derivatives} therefore gives
\begin{equation}
    \partial_{u_\beta}
    =
    -\frac{1}{\sqrt\theta}
    \frac{\partial}{\partial\xi_\beta}.
    \label{eq:app_u_chain_rule}
\end{equation}
Using \eqref{eq:app_inverse_ladder_relations},
\begin{align}
    \partial_{u_\beta}\varphi_{\bm n}
    &=-\frac{1}{\sqrt\theta}
    \frac{\partial\varphi_{\bm n}}{\partial\xi_\beta}
    \
    =
    \frac{1}{2\sqrt\theta}
    \left(\mathcal A_\beta^\dagger-\mathcal A_\beta\right)
    \varphi_{\bm n}.
    \label{eq:app_u_derivative_operator}
\end{align}
By \eqref{eq:app_ladder_action},
\begin{equation}
    \partial_{u_\beta}\varphi_{\bm n}
    =
    \frac{1}{2\sqrt\theta}
    \left[
        \sqrt{n_\beta+1}\,
        \varphi_{\bm n+\bm e_\beta}
        -
        \sqrt{n_\beta}\,
        \varphi_{\bm n-\bm e_\beta}
    \right].
    \label{eq:app_u_derivative_basis}
\end{equation}
It follows that
\begin{equation}
    \widehat{\mathcal C}_{u_\beta}
    =
    \frac{1}{2\sqrt\theta}
    \left(
        \hat a_\beta^\dagger-\hat a_\beta
    \right)
    .
    \label{eq:app_velocity_connector}
\end{equation}

\subsubsection{Direct vacuum check}

For the vacuum state,
\begin{align}
    \partial_{u_\beta}\sqrt W
    &=
    \sqrt W\,
    \partial_{u_\beta}
    \left(-\frac{c_\alpha c_\alpha}{4\theta}\right)
    =
    \frac{c_\beta}{2\theta}\sqrt W
    =
    \frac{\xi_\beta}{2\sqrt\theta}\sqrt W.
    \label{eq:app_u_vacuum_direct}
\end{align}
Since
\begin{equation}
    \mathcal A_\beta^\dagger\sqrt W
    =\xi_\beta\sqrt W,
\end{equation}
we recover
\begin{equation}
    \partial_{u_\beta}\varphi_{\bm0}
    =
    \frac{1}{2\sqrt\theta}
    \varphi_{\bm e_\beta},
    \label{eq:app_u_vacuum_check}
\end{equation}
which independently fixes both the sign and the factor $1/2$ in
\eqref{eq:app_velocity_connector}.

\subsection{Temperature connector}

The temperature dependence of $\varphi_{\bm n}$ enters both through the
normalization factor $\theta^{-d/4}$ and through the scaled variable
$\bm\xi$. At fixed $\bm v$,
\begin{equation}
    \partial_\theta\varphi_{\bm n}
    =
    -\frac{d}{4\theta}\varphi_{\bm n}
    -\frac{1}{2\theta}
    \xi_\alpha
    \frac{\partial\varphi_{\bm n}}{\partial\xi_\alpha}.
    \label{eq:app_theta_chain_rule}
\end{equation}
From \eqref{eq:app_inverse_ladder_relations},
\begin{align}
    \xi_\alpha\frac{\partial}{\partial\xi_\alpha}
    =
    \frac12
    (\mathcal A_\alpha+\mathcal A_\alpha^\dagger)
    (\mathcal A_\alpha-\mathcal A_\alpha^\dagger)
    =
    \frac12
    \left(
        \mathcal A_\alpha^2
        -\mathcal A_\alpha\mathcal A_\alpha^\dagger
        +\mathcal A_\alpha^\dagger\mathcal A_\alpha
        -(\mathcal A_\alpha^\dagger)^2
    \right)
=
    \frac12
    \left(
        \mathcal A_\alpha^2
        -(\mathcal A_\alpha^\dagger)^2
        -I
    \right),
    \label{eq:app_xi_partial_identity}
\end{align}
where $\mathcal A_\alpha\mathcal A_\alpha^\dagger
=\mathcal A_\alpha^\dagger\mathcal A_\alpha+I$ was used in the last step.
Substitution into \eqref{eq:app_theta_chain_rule} gives
\begin{align}
    \partial_\theta\varphi_{\bm n}
    &=
    -\frac{d}{4\theta}\varphi_{\bm n}
    +
    \frac{1}{4\theta}
    \sum_{\alpha=1}^{d}
    \left[
        (\mathcal A_\alpha^\dagger)^2
        -\mathcal A_\alpha^2
        +I
    \right]
    \varphi_{\bm n}
    =
    \frac{1}{4\theta}
    \sum_{\alpha=1}^{d}
    \left[
        (\mathcal A_\alpha^\dagger)^2
        -\mathcal A_\alpha^2
    \right]
    \varphi_{\bm n}.
    \label{eq:app_theta_derivative_operator}
\end{align}
The scalar term cancels because
\begin{equation}
    -\frac{d}{4\theta}I
    +\frac{1}{4\theta}\sum_{\alpha=1}^{d}I=0.
    \label{eq:app_scalar_cancellation}
\end{equation}
This cancellation depends essentially on the factor $\theta^{-d/4}$ in the
normalization of the Hermite functions.

Using \eqref{eq:app_ladder_action},
\begin{equation}
    \partial_\theta\varphi_{\bm n}
    =
    \frac{1}{4\theta}
    \sum_{\alpha=1}^{d}
    \Bigl[
        \sqrt{(n_\alpha+1)(n_\alpha+2)}\,
        \varphi_{\bm n+2\bm e_\alpha}
        -
        \sqrt{n_\alpha(n_\alpha-1)}\,
        \varphi_{\bm n-2\bm e_\alpha}
    \Bigr].
    \label{eq:app_theta_derivative_basis}
\end{equation}
Therefore,
\begin{equation}
    \widehat{\mathcal C}_\theta
    =
    \frac{1}{4\theta}
    \sum_{\alpha=1}^{d}
    \left[
        (\hat a_\alpha^\dagger)^2
        -\hat a_\alpha^2
    \right]
    .
    \label{eq:app_temperature_connector}
\end{equation}

\subsubsection{Direct vacuum check}

Direct differentiation of \eqref{eq:app_sqrtW} gives
\begin{align}
    \partial_\theta\sqrt W
    &=
    \left(
        -\frac{d}{4\theta}
        +\frac{c_\alpha c_\alpha}{4\theta^2}
    \right)
    \sqrt W
    =
    \frac{1}{4\theta}
    (\bm\xi^2-d)\sqrt W.
    \label{eq:app_theta_vacuum_direct}
\end{align}
On the other hand,
\begin{align}
    \sum_{\alpha=1}^{d}
    (\mathcal A_\alpha^\dagger)^2\sqrt W
    &=
    \sum_{\alpha=1}^{d}
    (\xi_\alpha^2-1)\sqrt W
    =
    (\bm\xi^2-d)\sqrt W,
    \label{eq:app_theta_vacuum_ladder}
\end{align}
and $\mathcal A_\alpha^2\sqrt W=0$. Hence
\begin{equation}
    \partial_\theta\varphi_{\bm0}
    =
    \frac{1}{4\theta}
    \sum_{\alpha=1}^{d}
    \sqrt2\,
    \varphi_{2\bm e_\alpha},
    \label{eq:app_theta_vacuum_check}
\end{equation}
which confirms \eqref{eq:app_temperature_connector}.

\subsection{Spacetime differential connector}

Let $x^\mu$ denote either time or a spatial coordinate. Since the realization
depends on $x^\mu$ only through the fields $u_\beta(x)$ and $\theta(x)$, the
chain rule gives
\begin{equation}
    \widehat{\mathcal C}_\mu
    =
    (\partial_\mu u_\beta)
    \widehat{\mathcal C}_{u_\beta}
    +
    (\partial_\mu\theta)
    \widehat{\mathcal C}_\theta.
    \label{eq:app_spacetime_connector_chain_rule}
\end{equation}
Combining \eqref{eq:app_velocity_connector} and
\eqref{eq:app_temperature_connector},
\begin{equation}
    \widehat{\mathcal C}_\mu
    =
    \frac{\partial_\mu u_\beta}{2\sqrt\theta}
    \left(
        \hat a_\beta^\dagger-\hat a_\beta
    \right)
    +
    \frac{\partial_\mu\theta}{4\theta}
    \sum_{\alpha=1}^{d}
    \left[
        (\hat a_\alpha^\dagger)^2
        -\hat a_\alpha^2
    \right]
    .
    \label{eq:app_spacetime_connector_final}
\end{equation}

For the transport operator
\begin{equation}
    \mathcal K_{\mathrm{tr}}
    =
    \partial_t
    +
    \left[
        u_\alpha
        +\sqrt\theta
        (\mathcal A_\alpha+\mathcal A_\alpha^\dagger)
    \right]
    \partial_\alpha,
\end{equation}
the pull-back under the realization map takes the form
\begin{equation}
    \widehat{\mathcal K}_{\mathrm{tr}}
    =
    \partial_t+\widehat{\mathcal C}_t
    +
    \left[
        u_\alpha I
        +\sqrt\theta
        (\hat a_\alpha+\hat a_\alpha^\dagger)
    \right]
    (\partial_\alpha+\widehat{\mathcal C}_\alpha).
    \label{eq:app_transport_pullback}
\end{equation}
Equation \eqref{eq:app_spacetime_connector_final} supplies the explicit
Hermite-function realization of all differential connectors appearing in
\eqref{eq:app_transport_pullback}.

\subsection{Independent consistency checks}

\subsubsection{Anti-Hermiticity}

Differentiating the orthonormality relation
\eqref{eq:app_basis_orthonormality} with respect to any realization parameter
$\lambda$ gives
\begin{equation}
    \langle\partial_\lambda\varphi_{\bm m},\varphi_{\bm n}\rangle
    +
    \langle\varphi_{\bm m},\partial_\lambda\varphi_{\bm n}\rangle
    =0.
    \label{eq:app_orthonormality_derivative}
\end{equation}
Thus the connector matrix must be anti-Hermitian,
\begin{equation}
    \widehat{\mathcal C}_\lambda^\dagger
    =-\widehat{\mathcal C}_\lambda.
    \label{eq:app_connector_antihermiticity}
\end{equation}
The derived formulas satisfy this property because
\begin{equation}
    (\hat a_\beta^\dagger-\hat a_\beta)^\dagger
    =
    -(\hat a_\beta^\dagger-\hat a_\beta),
\end{equation}
and
\begin{equation}
    \left[(\hat a_\alpha^\dagger)^2-\hat a_\alpha^2\right]^\dagger
    =
    -\left[(\hat a_\alpha^\dagger)^2-\hat a_\alpha^2\right].
\end{equation}

\subsubsection{Dimensional analysis}

The abstract ladder operators are dimensionless. Since $\sqrt\theta$ has the
dimension of velocity,
\begin{equation}
    [\widehat{\mathcal C}_{u_\beta}]
    =[u]^{-1},
    \qquad
    [\widehat{\mathcal C}_\theta]
    =[\theta]^{-1},
\end{equation}
as required for derivatives with respect to $u_\beta$ and $\theta$.
Furthermore,
\begin{equation}
    [\widehat{\mathcal C}_\mu]=[\partial_\mu],
\end{equation}
so that each sum $\partial_\mu+\widehat{\mathcal C}_\mu$ is dimensionally
homogeneous.

\subsubsection{Low-level matrix elements}

The first few nonzero actions provide simple tests of all occupation-number
factors. In one Cartesian direction,
\begin{align}
    \widehat{\mathcal C}_{u}|0\rangle
    &=
    \frac{1}{2\sqrt\theta}|1\rangle,
    \
    \widehat{\mathcal C}_{u}|1\rangle
    =
    \frac{1}{2\sqrt\theta}
    \left(\sqrt2|2\rangle-|0\rangle\right),
    \
    \widehat{\mathcal C}_{u}|2\rangle
    =
    \frac{1}{2\sqrt\theta}
    \left(\sqrt3|3\rangle-\sqrt2|1\rangle\right),
    \label{eq:app_low_u_states}
\end{align}
and
\begin{align}
    \widehat{\mathcal C}_{\theta}|0\rangle
    &=
    \frac{\sqrt2}{4\theta}|2\rangle,
    \
    \widehat{\mathcal C}_{\theta}|1\rangle
    =
    \frac{\sqrt6}{4\theta}|3\rangle,
    \
    \widehat{\mathcal C}_{\theta}|2\rangle
    =
    \frac{1}{4\theta}
    \left(2\sqrt3|4\rangle-\sqrt2|0\rangle\right).
    \label{eq:app_low_theta_states}
\end{align}
The off-diagonal matrix elements obey
\begin{equation}
    \langle m|\widehat{\mathcal C}_\lambda|n\rangle
    =
    -
    \langle n|\widehat{\mathcal C}_\lambda|m\rangle^*,
\end{equation}
in agreement with anti-Hermiticity.

\subsubsection{Parity selection rules}

The velocity connector changes the total occupation number by one,
\begin{equation}
    \Delta|\bm n|=\pm1,
\end{equation}
whereas the temperature connector changes it by two,
\begin{equation}
    \Delta|\bm n|=\pm2.
\end{equation}
Thus variations of the local mean velocity mix opposite-parity Hermite sectors,
while variations of the local temperature preserve total Hermite parity. These
selection rules follow independently from translation and dilation of the local
Gaussian basis, respectively.

\subsection{Summary}

For the normalized Hermite functions
\begin{equation}
    \varphi_{\bm n}
    =
    \sqrt W\,
    \frac{H_{\bm n}(\bm\xi/\sqrt2)}
    {\sqrt{2^{|\bm n|}\bm n!}},
\end{equation}
the elementary differential connectors are
\begin{align}
    \widehat{\mathcal C}_{u_\beta}
    &=
    \frac{1}{2\sqrt\theta}
    (\hat a_\beta^\dagger-\hat a_\beta),
    \
    \widehat{\mathcal C}_\theta
    =
    \frac{1}{4\theta}
    \sum_{\alpha=1}^{d}
    \left[
        (\hat a_\alpha^\dagger)^2-\hat a_\alpha^2
    \right].
\end{align}
The factors $1/(2\sqrt\theta)$ and $1/(4\theta)$, the relative signs, and the
absence of a scalar term in the temperature connector are fixed by the exact
normalization of the Hermite-function basis. They are confirmed independently
by direct vacuum differentiation, orthonormality, anti-Hermiticity,
dimensional analysis, and low-level matrix elements.




\begin{thebibliography}{99}

\bibitem{LebowitzFrischHelfand1960}
J.~L. Lebowitz, H.~L. Frisch, and E. Helfand,
``Nonequilibrium Distribution Functions in a Fluid,''
\emph{Physics of Fluids} \textbf{3} (1960), 325--338.
\href{https://doi.org/10.1063/1.1706037}{doi:10.1063/1.1706037}.

\bibitem{UhlenbeckOrnstein1930}
G.~E. Uhlenbeck and L.~S. Ornstein,
``On the Theory of the Brownian Motion,''
\emph{Physical Review} \textbf{36} (1930), 823--841.
\href{https://doi.org/10.1103/PhysRev.36.823}{doi:10.1103/PhysRev.36.823}.

\bibitem{Kramers1940}
H.~A. Kramers,
``Brownian Motion in a Field of Force and the Diffusion Model of Chemical Reactions,''
\emph{Physica} \textbf{7} (1940), 284--304.
\href{https://doi.org/10.1016/S0031-8914(40)90098-2}{doi:10.1016/S0031-8914(40)90098-2}.

\bibitem{Risken1989}
H. Risken,
\emph{The Fokker--Planck Equation: Methods of Solution and Applications},
2nd ed., Springer Series in Synergetics, Vol.~18, Springer, Berlin (1989).
\href{https://doi.org/10.1007/978-3-642-61544-3}{doi:10.1007/978-3-642-61544-3}.

\bibitem{GradHermite1949}
H. Grad,
``Note on $N$-Dimensional Hermite Polynomials,''
\emph{Communications on Pure and Applied Mathematics} \textbf{2} (1949), 325--330.
\href{https://doi.org/10.1002/cpa.3160020402}{doi:10.1002/cpa.3160020402}.

\bibitem{GradKinetic1949}
H. Grad,
``On the Kinetic Theory of Rarefied Gases,''
\emph{Communications on Pure and Applied Mathematics} \textbf{2} (1949), 331--407.
\href{https://doi.org/10.1002/cpa.3160020403}{doi:10.1002/cpa.3160020403}.

\bibitem{ChapmanCowling1970}
S. Chapman and T.~G. Cowling,
\emph{The Mathematical Theory of Non-Uniform Gases},
3rd ed., Cambridge University Press, Cambridge (1970).

\bibitem{Liboff2003}
R.~L. Liboff,
\emph{Kinetic Theory: Classical, Quantum, and Relativistic Descriptions},
3rd ed., Springer, New York (2003).
\href{https://doi.org/10.1007/b97467}{doi:10.1007/b97467}.

\bibitem{Fock1932}
V.~A. Fock,
``Konfigurationsraum und zweite Quantelung,''
\emph{Zeitschrift f{\"u}r Physik} \textbf{75} (1932), 622--647.
\href{https://doi.org/10.1007/BF01344458}{doi:10.1007/BF01344458}.

\bibitem{Berezin1966}
F.~A. Berezin,
\emph{The Method of Second Quantization},
Academic Press, New York (1966).

\bibitem{Hall2013}
B.~C. Hall,
\emph{Quantum Theory for Mathematicians},
Graduate Texts in Mathematics, Vol.~267, Springer, New York (2013).
\href{https://doi.org/10.1007/978-1-4614-7116-5}{doi:10.1007/978-1-4614-7116-5}.

\bibitem{Folland1989}
G.~B. Folland,
\emph{Harmonic Analysis in Phase Space},
Annals of Mathematics Studies, Vol.~122, Princeton University Press, Princeton, NJ (1989).

\bibitem{Doi1976}
M. Doi,
``Second Quantization Representation for Classical Many-Particle System,''
\emph{Journal of Physics A: Mathematical and General} \textbf{9} (1976), 1465--1477.
\href{https://doi.org/10.1088/0305-4470/9/9/008}{doi:10.1088/0305-4470/9/9/008}.

\bibitem{Peliti1985}
L. Peliti,
``Path Integral Approach to Birth--Death Processes on a Lattice,''
\emph{Journal de Physique} \textbf{46} (1985), 1469--1483.
\href{https://doi.org/10.1051/jphys:019850046090146900}{doi:10.1051/jphys:019850046090146900}.

\bibitem{JennyTorrilhonHeinz2010}
P. Jenny, M. Torrilhon, and S. Heinz,
``A Solution Algorithm for the Fluid Dynamic Equations Based on a Stochastic Model for Molecular Motion,''
\emph{Journal of Computational Physics} \textbf{229} (2010), 1077--1098.
\href{https://doi.org/10.1016/j.jcp.2009.10.008}{doi:10.1016/j.jcp.2009.10.008}.

\bibitem{GorjiJenny2014}
M.~H. Gorji and P. Jenny,
``An Efficient Particle Fokker--Planck Algorithm for Rarefied Gas Flows,''
\emph{Journal of Computational Physics} \textbf{262} (2014), 325--343.
\href{https://doi.org/10.1016/j.jcp.2013.12.046}{doi:10.1016/j.jcp.2013.12.046}.

\bibitem{GorjiTorrilhon2021}
M.~H. Gorji and M. Torrilhon,
``Entropic Fokker--Planck Kinetic Model,''
\emph{Journal of Computational Physics} \textbf{430} (2021), 110034.
\href{https://doi.org/10.1016/j.jcp.2020.110034}{doi:10.1016/j.jcp.2020.110034}.

\bibitem{LiuEtAl2021}
J.-P. Liu, H.~{\O}. Kolden, H.~K. Krovi, N.~F. Loureiro, K. Trivisa, and A.~M. Childs,
``Efficient Quantum Algorithm for Dissipative Nonlinear Differential Equations,''
\emph{Proceedings of the National Academy of Sciences} \textbf{118} (2021), e2026805118.
\href{https://doi.org/10.1073/pnas.2026805118}{doi:10.1073/pnas.2026805118}.

\bibitem{bravyi2025quantumsimulationnoisyclassical}
S. Bravyi, R. Manson-Sawko, M. Zayats, and S. Zhuk,
``Quantum simulation of a noisy classical nonlinear dynamics,''
arXiv:2507.06198 [quant-ph] (2025).
\url{https://arxiv.org/abs/2507.06198}.

\bibitem{bravyi2026quantumalgorithmsstochasticnonlinear}
S. Bravyi, A. Byrne, M. Zayats, and S. Zhuk,
``Quantum algorithms for stochastic nonlinear differential equations,''
arXiv:2606.08349 [quant-ph] (2026).
\url{https://arxiv.org/abs/2606.08349}.

\end{thebibliography}
\end{document}